\documentclass[10pt, reqno]{amsart}
\usepackage{amsmath, amssymb}
\usepackage[bitstream-charter]{mathdesign}
\usepackage[T1]{fontenc}

\usepackage{mathrsfs}
\usepackage[arrow,matrix,curve,cmtip,ps]{xy}

\usepackage{amsthm}

\usepackage{float}
\usepackage{hyperref}
\usepackage{graphics}
\usepackage{graphicx}
\usepackage{verbatim}
\usepackage{multirow}
\usepackage{tikz}
\usepackage{enumerate}
\usepackage{eufrak}
\newcommand*\mycirc[1]{%
  \begin{tikzpicture}
    \node[draw,circle,inner sep=1pt] {#1};
  \end{tikzpicture}}

\allowdisplaybreaks

\newtheorem{theorem}{Theorem}[section]
\newtheorem{lemma}[theorem]{Lemma}
\newtheorem{proposition}[theorem]{Proposition}
\newtheorem{corollary}[theorem]{Corollary}

\newtheorem*{theorem*}{Theorem}
\theoremstyle{remark}

\newtheorem{definition}[theorem]{Definition}

\numberwithin{equation}{section}

\newcommand{\R}{\mathbb{R}}

\begin{document}
\title{The N-player war of attrition in the limit of infinitely many players}

\author{Peter Helgesson, Bernt Wennberg}

\address{Department of Mathematics \\ Chalmers University of Technology \\ SE41296 Gothenburg \\ Sweden}
\email{helgessp@chalmers.se, wennberg@chalmers.se}

\date{\today}



\begin{abstract}

The \textit{War of Attrition} is a classical game theoretic model that was first introduced to mathematically describe certain non-violent animal behavior. The original setup considers two participating players in a one-shot game competing for a given prize by waiting. This model has later been extended to several different models allowing more than two players. One of the first of these $N$-player generalizations was due to J. Haigh and C. Cannings in \cite{HaigCannings} where two possible models are mainly discussed; one in which the game starts afresh with new strategies each time a player leaves the game, and one where the players have to stick with the strategy they chose initially. The first case is well understood whereas, for the second case, much is still left open.\\
\indent In this paper we study the asymptotic behavior of these two models as the number of players tend to infinity and prove that their time evolution coincide in the limit. We also prove new results concerning the second model in the $N$-player setup.

\end{abstract}

\maketitle

\tableofcontents

\section{Introduction}
Game theory has ever since the pioneering works by J. von Neumann developed in to an important tool in the study of various areas of research such as economical science, computer science, political science, biology, social science and even in philosophy. A common point of view is that game theory constitutes a theory of rational and strategic decision making describing how rational players would optimize their play, often in terms of Nash-equilibrium. During the years especially economical science has earned a lot of success applying game theory in various situations, and this has resulted in several Nobel-prizes. The latest of these was given to Alvin E. Roth and Loyd S. Shapley in 2012 "for the theory of stable allocations and the practice of market design". However, when applying game theory to problems in biology and animal behavior it is obvious that the common view point of having rational players is insufficient. Even though many situations in biology, in principle, could be described as some kind of game, one can not consider animals as being actively rational. One rather expect animal behavior to be in agreement with game theory as a consequence of natural selection in evolution. In 1973 in \cite{MaynardPrice} J. Maynard Smith and G. R. Price introduced the notion of \textit{Evolutionary Stable Strategy}, in short ESS, that was to take the same place in game theoretic biology as Nash-equilibrium had had in game theoretic economy. The ESS serves as the natural candidate for what type of animal behavior that evolution eventually would produce by natural selection. In 1974, published in \cite{Maynard}, J. Maynard Smith developed a game theoretic, non-violent, conflict scenario called \textit{War of Attrition} to describe potential animal behavior in e.g. territorial competition. The model considers two players competing for one single prize $V \in \R_+$ by waiting. The cost of waiting is modeled as proportional to the duration of the game, and it is payed in the same amounts by both parts when the first player decides to leave. The remaining player wins the game and collects the prize $V$. In \cite{BishopCannings} it was proven by D.T. Bishop and C. Cannings that the war of attrition has one unique mixed ESS given by choosing waiting time at random from an exponential probability distribution having mean $V$. In 1999 John Maynard Smith, together with E. Mayr and G. C. Williams, was honored with the Crafoord prize for his work in evolutionary biology in connection with game theory.\\
\indent In 1989 J. Haigh and C. Cannings in \cite{HaigCannings} generalized the two player model of the war of attrition to models involving several players. One could of course think of many ways of constructing such generalizations, but the ones considered in \cite{HaigCannings} are probably the most natural extensions. In this text we will refer to these models as the \textit{dynamic model} and the \textit{static model}\footnote{In \cite{HaigCannings} the dynamic model is called Model C and the static model is called Model D.}. The $N$-player dynamic model of the war of attrition is a repetitive game in $N-1$ rounds in which one player drops out of the game in each round until there is only two players left in the final round. Between the rounds the remaining players are allowed to change their strategies for the next round. The dynamic model is well understood and the existence and uniqueness of an ESS is proven in \cite{HaigCannings} under very general conditions.
\\ \indent In the $N$-player static model of the war of attrition all participating players choose their waiting time at the beginning of the game. Each of them are then bound to stick to their chosen waiting time. Hence the static model is a one-shoot game, i.e. the outcome of the game is known as soon as all players have made their choice. In contrast to the dynamic model far less is know about how to play the static model. In \cite{HaigCannings} it is proven by specific examples that the static model admit a unique ESS in some cases while in other cases it does not, and much is left open.\\
\indent The war of attrition has through time developed into one of the most classic game theoretic models. It has been studied from a different interesting point of view in \cite{Lundh}.

\section{Preliminaries and Introductory Results}

We begin with a heuristic discussion. For the simplest setup of the war of attrition, from now on WA, (see \cite{Maynard}) we consider a two player game in which the contestants are competing for a prize $V > 0$ by waiting. There is a cost connected to the duration $t$ of the game modelled linearly as $-t$. The game ends once one of the players decide to withdraw by paying the collected time cost and leave the price $V$ to the opponent player, who also pays the time cost. If we name the players by $x$ and $y$, and their corresponding waiting times by $\tau_x$ and $\tau_y$, we get the WA pay-off function for player $x$ as:
\begin{equation}
    J_x(\tau_x, \tau_y) := \left \{
    \begin{array}{rl}
        V -\tau_y,& \text{ if } \tau_x > \tau_y\\
        -\tau_x,& \text{ if } \tau_x < \tau_y.
    \end{array} \right.
\end{equation}

\noindent In the case of equal waiting times we define
\begin{equation}
    J_x(\tau,\tau) := \frac{V}{2} - \tau,\hspace{1.3cm} \forall \tau \in [0,\infty).
\end{equation}

\noindent It is clear that this setup of the game can not have a pure strategy ESS, or even a pure strategy Nash-equilibrium, since if there were such a strategy it would be given by a fixed waiting time $\bar{\tau}$. It would therefore always be possible to brake this strategy by waiting just a bit longer than $\bar{\tau}$. However, according to \cite{BishopCannings}, there is a unique mixed ESS given by letting $\tau \sim \text{exp}(1/V)$, i.e. letting $\tau$ be randomly distributed with an exponential density of mean $V$. As mentioned in the introduction, in \cite{HaigCannings} J. Haigh and C. Cannings generalized the above two player setup of the WA to two different models allowing $N$ players; one repetitive game, the dynamic model, and one one-shot game, the static model. In both cases one consider a sequence $\left \{ V_k \right \}_{k=1}^N$ of positive numbers representing the prizes that the $N$ players are competing for. In this text we will assume this prize sequence to be an increasing sequence of real positive numbers, i.e. $0 \leq V_1 < V_2 < ... < V_N$.\\
\indent The dynamic $N$-player model is divided into $N-1$ distinct rounds. In the beginning of the first round all the $N$ players, independently of each other, choose their waiting times. Then the players wait and the contestant having the least waiting time $\tau_{(1)}^1$ leaves the game by receiving the prize $V_1$ and paying the time cost $\tau_{(1)}^1$. The remaining $N-1$ players also pay the cost $\tau_{(1)}^1$ and proceed into the second round where the game starts afresh and proceeds as in the first round, playing for the prize $V_2$ instead. The game goes on until the $(N-1)$'th player leaves in the final round by receiving $V_{N-1}$ and paying $\tau_{(1)}^{N-1}$, thus leaving the final player left to claim the prize $V_N$ for a total cost of $\tau_{(1)}^1 + \tau_{(1)}^2 + ... + \tau_{(1)}^{N-1}$.\\
\indent It was proven in \cite{HaigCannings} that there exists a unique mixed ESS for each round $k = 1, 2, ..., N-1$ in the above dynamic model by choosing waiting time according to an exponential distribution with mean $(N-k)(V_{k+1} - V_k)$. In what follows we will use this result to investigate $N$-player limit of the dynamic model in a "sketchy manner". Given the increasing sequence of prizes $\left \{ V_k \right \}_{k=1}^N$ we associate a piecewise linear function, $V^N(x)$ on $x \in [0,1]$, by declaring $V^N(0) := 0$ and $V^N(k/N) := V_k$ so that every pair $\left\{ (k/N, V^N(k/N)), ((k+1)/N, V^N((k+1)/N) \right\}$ is joined together by a line segment. It is clear that the function $V^N$ may have a very bad behaviour in the limit as $N\longrightarrow \infty$. For instance if $V_k := k$ we would get an a.e. unbounded function in the limit. However, if we suppose that the prize sequence is such that $V^N \longrightarrow V \in \mathcal{C}^1([0,1])$ as $N\longrightarrow \infty$ the dynamic $N$-player WA will have meaning in the limit\footnote{Indeed, this can be accieved by starting from an increasing function $g \in \mathcal{C}^1([0,1])$ and simply define the sequence $\left( V_k \right)_{k=1}^N$ as $V_k := g(k/N)$.} and we can investigate the limiting behaviour. If we denote the density function of the mixed ESS of round $k$ by $f^N_k$, and let $k = \lceil qN \rceil$ for some fixed $q \in [0,1]$ ($\lceil x \rceil := \mbox{the smallest integer bigger than } x$), we have that
\begin{align}
    f^N_k(\tau) &:= \frac{1}{(N-k)\left( V\left( \frac{k+1}{N}\right) - V\left( \frac{k}{N}\right) \right)} \text{exp}\left\{-\frac{\tau}{(N-k)\left( V\left( \frac{k+1}{N}\right) - V\left( \frac{k}{N}\right) \right)}\right\}\nonumber \\
        & = \frac{1}{(1-\frac{k}{N})\left( \frac{V\left( \frac{k+1}{N}\right) - V\left( \frac{k}{N}\right)}{1/N} \right)} \text{exp} \left\{ -\frac{\tau}{(1-\frac{k}{N})\left( \frac{V\left( \frac{k+1}{N}\right) - V\left( \frac{k}{N}\right)}{1/N} \right)} \right\} \longrightarrow \nonumber\\
        & \longrightarrow \frac{1}{(1-q)V'(q)}e^{-\frac{\tau}{(1-q)V'(q)}} =: f(\tau),
        \label{lim1}
\end{align}
\noindent as $N \longrightarrow \infty$. The number $q$ represents the fraction of players that, at the moment, have left the game. Of course, in this setup $q$ depends on the time $t\in \mathbb{R}_+$ and we would like to analyze its time evolution and how it relates to $V(x)$. For this we introduce the \textit{mean field density function} $m(t,\tau)$ describing the fraction of players still left at time $t$ after the game has started, with chosen waiting times $\tau$. Thus $m$ will lose mass as players are quitting according to
\begin{equation}
    \int_0^{\infty} m(t,\tau) d\tau = 1-q(t),
\end{equation}

\noindent and since the $\tau$-marginal should be exponentially distributed like (\ref{lim1}) for every $q$ we get
\begin{equation}
    m(t,\tau) := \frac{1}{V'(q(t))}e^{-\frac{\tau}{(1-q(t))V'(q(t))}}.
\end{equation}

\noindent Thinking of $m(t,.)$ as an approximation of the distribution of players in the $N$-player game at time $t$, let $n_0^t$ denote the number of players in the vicinity of $\tau = 0$ at $t$ (i.e. having $\tau \in [0,d\tau]$). Then
\begin{equation*}
    m(t,0)d\tau \approx \frac{n_0^t}{N}.
\end{equation*}

\noindent On the other hand, since $t$ and $\tau$ are of the same time scale and $m$ is continuous we should also have that
\begin{equation}
    m(t,0)dt \approx \frac{n_0^t}{N}
\end{equation}

\noindent and therefore $m(t,0)d\tau = m(t,0)dt$. This means that the number of players having waiting times $\tau \in [0,d\tau]$ equals to the number of players that will leave the game in the interval $[t, t+dt]$. By this we get that
\begin{equation}
    q(t) = \int_0^t m(s,0)ds = \int_0^t \frac{1}{V'(q(s))} ds,
\end{equation}

\noindent which in turn yields a differential equation for the dynamics of $q(t)$ as
\begin{equation}
    \frac{d}{dt}V(q(t)) = V'(q(t))\dot{q}(t) = 1.
    \label{diff1}
\end{equation}

\noindent Since $q(0)=0$, equation (\ref{diff1}) suggests that
\begin{equation}
    V(q(t)) - V(0) = t
    \label{time}
\end{equation}

\noindent and since the game by definition will end when all the players are out, i.e. when $q(t)=1$, we get by (\ref{time}) that the total duration $T$ of the game is given by the formula
\begin{equation}
    T := V(1) - V(0).
    \label{TN}
\end{equation}

\noindent In the following lemma we show that this result is consistent with the corresponding result one would get from \cite{HaigCannings} in the limit of $N$.

\begin{lemma}
Let $T_N := \sum_{k=1}^{N-1} \mathbb{E}[\tau_{(1)}^k]$, where each $\tau_{(1)}^k$ is a random variable with density given by the first order statistics of $(N-k+1)$ number of exponentially distributed random variables with parameter $(N-k) \left( V\left(\frac{k+1}{N}\right) - V\left(\frac{k}{N}\right) \right)$. Then
\begin{equation*}
    T_N \longrightarrow V(1)-V(0) = T,
\end{equation*}
as $N\longrightarrow \infty$.
\label{lemma1}
\end{lemma}

\begin{proof}
If $X_1, ..., X_N \sim \text{exp}(\lambda)$ and $X_{(1)} := \text{min}\{ X_1, ..., X_N \}$ we have that $\mathbb{E} [X_{(1)}] = 1/(\lambda N)$. Thus, for the sequence $(\tau_{(1)}^k)_{k=1}^{N-1}$ we have
\begin{equation*}
    \mathbb{E}[\tau_{(1)}^k] = \frac{N-k}{N-k+1}\left( V\left( \frac{k+1}{N} \right) - V\left( \frac{k}{N} \right) \right)
\end{equation*}

\noindent so
\begin{align*}
    T_N &= \sum_{k=1}^{N-1} \frac{N-k}{N-k+1}\left( V\left( \frac{k+1}{N} \right) - V\left( \frac{k}{N} \right) \right)\\
        &= \sum_{k=1}^{N-1} \frac{1-\frac{k}{N}}{1-\frac{k}{N}+\frac{1}{N}}\left( V\left( \frac{k+1}{N} \right) - V\left( \frac{k}{N} \right) \right)
        \longrightarrow V(1)-V(0)
\end{align*}
\noindent as $N\longrightarrow \infty$, since the quotient $(1-k/N)/(1-k/N-1/N) \longrightarrow 1$ independently of $k$, and we end up with a telescopic sum.
\end{proof}
In the dynamic $N$-player WA the quantity $T_N$ is a natural measure of the expected total time a typical game will last. In the first round all players chose their waiting times according to an exponential distribution with mean $(N-1)(V_2-V_1)$. If $\tau_i^1$ denotes the waiting time of player $i$ in the first round, we get a sequence of waiting times $\tau_1^1, \tau_2^1, ..., \tau_N^1$. The first round of the game will therefor last for a time $\tau_{(1)}^1 := \text{min}(\tau_1^1, \tau_2^1, ..., \tau_N^1)$, that is, the first order statistics of the waiting time sequence. In this case, since $\tau_1^1, \tau_2^1, ..., \tau_N^1$ are i.i.d. and exponentially distributed, it is well known that $\mathbb{E}[\tau_{(1)}^1] = (N-1)(V_2-V_1)/N$ (see e.g. \cite{David}). After the first round the game starts afresh and the $(N-1)$ players (independently of the previous round) chose their new waiting times $\tau_1^2, \tau_2^2, ..., \tau_{N-1}^2$ according to an exponential distribution, now having mean $(N-2)(V_3-V_2)$. The expected time of round two is again derived by first order statistics. The game continues like this until the $(N-1)$'th player leaves and the final prize $V_N$ is collected by the "winner". Thus the expected duration of the $N$-player game is precisely given by $T_N$. The result of Lemma \ref{lemma1} indicates consistency between the heuristic arguments that led to (\ref{TN}) and \cite{HaigCannings}.\\
\indent Given an increasing $\mathcal{C}^1$-prize function on the unit interval it is an easy task to solve the ode in (\ref{diff1}) and thereby derive the mean field density of players $m(t,\tau)$. Below we have included some numerical results of $m(t,\tau)$ and $q(t)$ for some different choices of prize function.
\begin{center}
\includegraphics[scale=0.4]{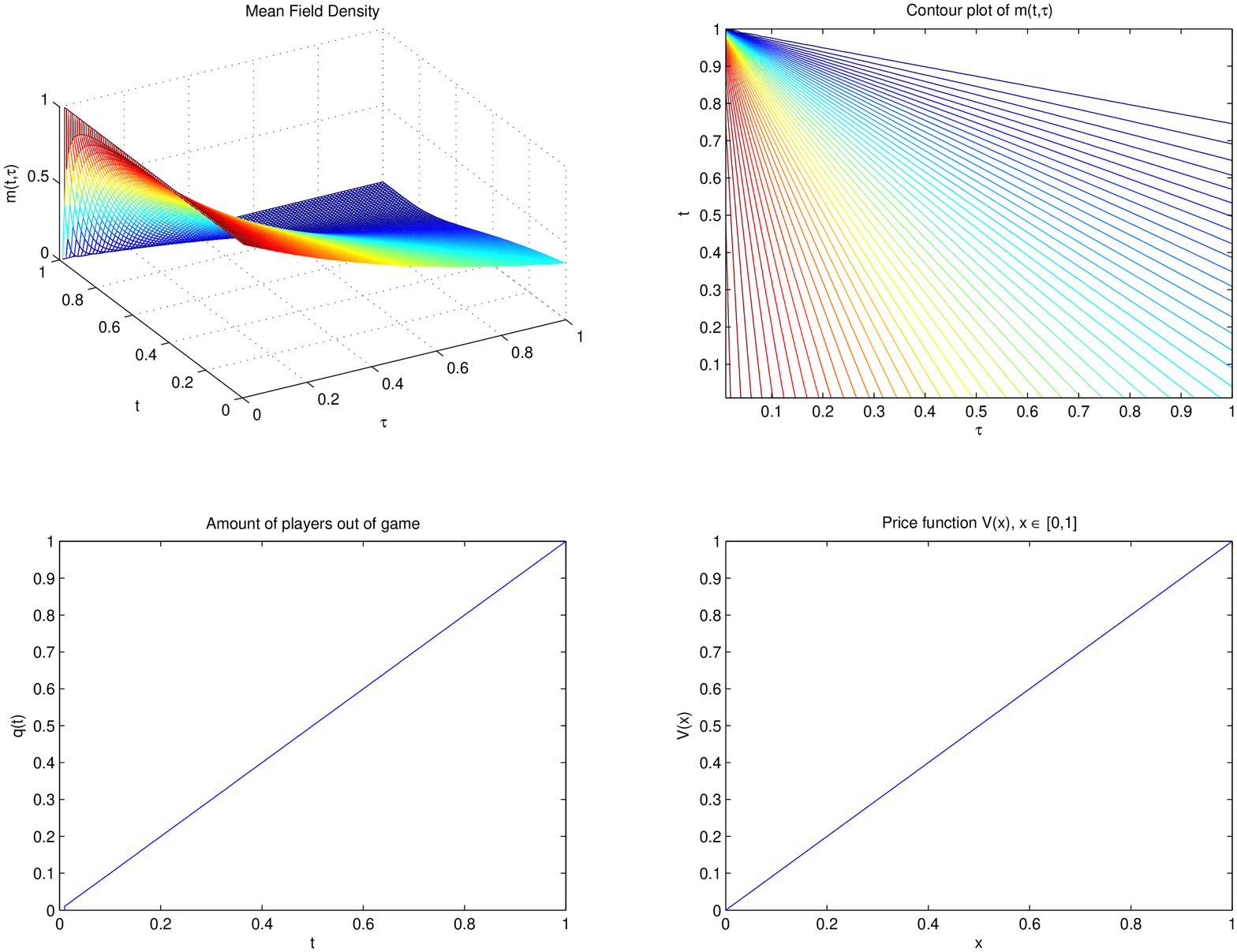}
\end{center}
\begin{center}
\includegraphics[scale=0.4]{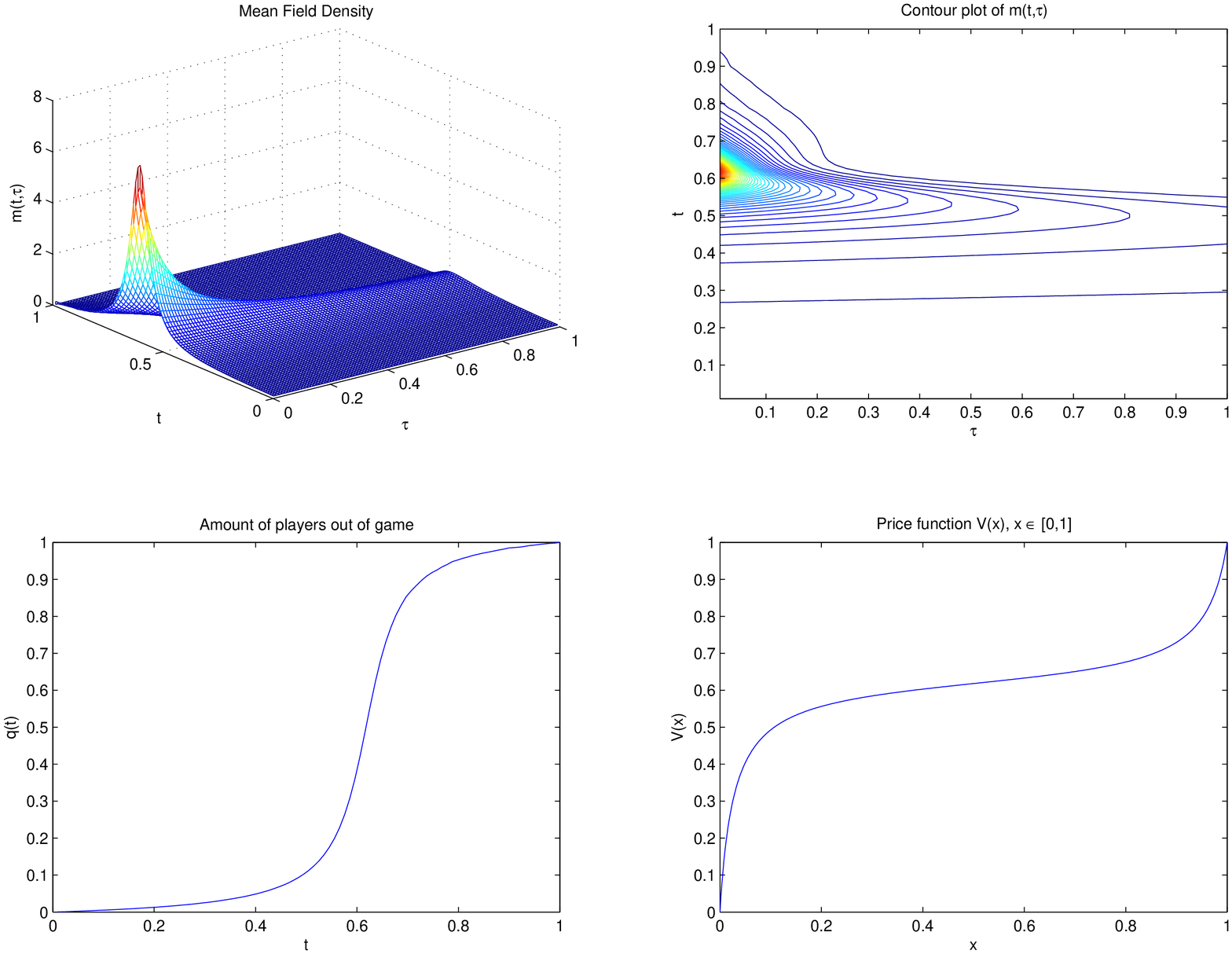}
\end{center}
The results above concerning the time evolution of $q$ are of course based on non rigorous arguments, but the main conclusion of $q(t)$ being the inverse function of $V(x)$ makes sense from a game theoretic point of view. To be more precise; let $\bar{t} \in [0,T]$ be a fixed point of time representing any pure strategy in the limiting dynamic WA with prize function $V(x)$. If all players in the game are playing according to $q(t)$ the payoff of playing any pure strategy $\bar{t}$ would be $V(q(\bar{t})) - \bar{t} = V(0)$ regardless of $\bar{t}$. In other words, $q(t)$ represents a Nash-equilibrium in the limit of infinitely many players. In the next section we investigate this more thoroughly.

\section{Convergence in the $N$-player limit of the dynamic model}
\label{section3}
In this section we consider the dynamic $N$-player generalization of the WA according to \cite{HaigCannings} and its behaviour as the number of players grows to infinity. We will assume the sequence $\{V_k \}_{k = 1}^N$ to be positive and strictly increasing and consider an $(N-1)$-round game (playing for $V_k$ in the $k$'th round) in which each round starts afresh once a player drops out. As stated in the introduction, we know that the dynamic model of the game has a unique mixed ESS given by a certain exponential distribution in each round. We assume that there is an increasing $\mathcal{C}^1$-function defined on the unit interval, denoted by $V(x)$, such that $V(0) = 0$ and $V_k := V(k/N)$.\\
\indent Since the mixed strategy ESS is an exponential distribution we may consider the evolution of the game as a continuous time Markov chain $\left( X(t) \right)_{t \geq 0}$, where $X(t)$ counts the total fraction of players that have decided to leave at $t$. Since a player that left the game never returns $X(t)$ will be a pure birth process. More specifically; if all the players left in the game after the first $k-1$ rounds play according to the ESS the time it takes to play the $k$'th round is given by the random variable $T_k := \text{min} (\tau^k_1, ..., \tau^k_{N-k+1})$, where $\tau^k_1, ..., \tau^k_{N-k+1}$ are i.i.d. exponentially distributed with mean $(N-k)(V_{k+1} - V_k)$. Therefore
\begin{equation}
    T_k \sim \text{exp}\left( \frac{N-k+1}{(N-k)\left( V_{k+1} - V_k \right)} \right)
    \label{T_k}
\end{equation}
and if we let $\lambda_k := (N-k+1)/((N-k)\left( V_{k+1} - V_k \right))$, for $k = 1, ..., N-1$, we have the finite birth process below describing the time evolution of the game as players are quitting:
\begin{equation}
    \mycirc{0} \xrightarrow{\lambda_1} \mycirc{1} \xrightarrow{\lambda_2} \mycirc{2} \xrightarrow{\lambda_3} ... \xrightarrow{\lambda_N} \mycirc{\textit{N}}
    \label{birthprocess}
\end{equation}
We define $\lambda_N := 0$ since the game ends as soon as $N-1$ players has left. Now, consider the stochastic jump process
\begin{equation}
    X(t) = \sum_{k=1}^{N-1}\frac{1}{N} \mathbb{I}_{ \{T_1 + ... + T_k \leq t \} }.
    \label{x}
\end{equation}
 Then $X(t)$ is a continuous time Markov process having the value $k/N$ during the $(k+1)$'th round of the game. We are interested in the limit $\lim_{N \rightarrow \infty} \mathbb{E} [X(t)]$ and to prove convergence towards $q(t)$ (see \ref{diff1}). In order to find a closed form expression of the expectation of $X(t)$ we use standard methods from continuous time Markov chain theory (see e.g. \cite{Anderson}). To the pure birth process in (\ref{x}) we have the associated intensity matrix
\begin{equation}
    \textbf{Q} :=
    \begin{pmatrix}
        -\lambda_1 & \lambda_1 & 0 & \cdots & \cdots & 0 \\
        0 & -\lambda_2 & \lambda_2 & 0 & \cdots & 0 \\
        \vdots & \vdots & \vdots & \ddots & \ddots & \vdots \\
        0 & 0 & 0 & \cdots & -\lambda_{N-1} & \lambda_{N-1} \\
        0 & 0 & 0 & 0 & \cdots & 0
    \end{pmatrix}.
\end{equation}
The Chapman-Kolomogorov equations state that
\begin{equation}
    \left\{
    \begin{array}{l}
        \frac{d\textbf{P}}{dt}(t) = \textbf{Q} \cdot \textbf{P}(t), \hspace{0.5cm} t \geq 0 \\
        \textbf{P}(0) = I,
    \end{array}
    \right.
    \label{chapmankolomogorov}
\end{equation}
where $\textbf{P}(t) = \left( p_{ij}(t) \right) \in \mathbb{R}^{N \times N}$ is the matrix of transition probabilities from state $i$ to state $j$ in (\ref{x}), at time $t$. Henceforth, for the sake of simplicity, we will assume that $\lambda_1 \neq \lambda_2 \neq ... \neq \lambda_N$. Allowing equalities would make the analysis much more involved, but it would not contribute to any more interesting results. Solving (\ref{chapmankolomogorov}) by hand is tedious but straight forward, and it is possible to find a closed form expression of the transition matrix $\textbf{P}(t)$ (see Appendix \ref{appendix1}). The state probabilities $p_i(t) := \mathbb{P} \{ X(t) = (i-1)/N \}$, collected in $\textbf{p}(t) = (p_1(t), ..., p_N(t))$, can be computed using the relation $\textbf{p}(t) = \textbf{p}(0) \cdot \textbf{P}(t)$ together with the initial condition $\textbf{p}(0) = (1, 0, ..., 0)$. The result is
\begin{equation}
    p_i(t) =
        \left\{
        \begin{array}{ll}
            e^{-\lambda_1 t}, & \text{if} \hspace{0.5cm} i=1 \\
            \sum_{l = 1}^{i} \frac{\prod_{k = 1}^{i-1} \lambda_k }{\prod_{k=1, k\neq l}^{i}\left(  \lambda_k - \lambda_l\right)} \cdot e^{-\lambda_l t}, & \text{if} \hspace{0.5cm} 1 < i \leq N.
        \end{array}
        \right.
\end{equation}
Therefore
\begin{equation}
    \mathbb{E}[ X(t)] = \sum_{i=0}^{N-1} \frac{i}{N} \cdot p_{i+1}(t)  = \sum_{i=1}^{N-1} \frac{i}{N} \sum_{l = 1}^{i+1} \frac{\prod_{k = 1}^{i} \lambda_k}{\prod_{k=1, k\neq l}^{i+1}\left(  \lambda_k - \lambda_l\right)} \cdot e^{-\lambda_l t}.
\end{equation}
In order to proceed and to understand this complicated expression the following lemma is useful:
\begin{lemma}
    Let $\left\{ \lambda_i \right\}_{i=1}^{n}$ be a sequence of positive and distinct real numbers. If $f_i(t) = \lambda_i e^{-\lambda_i t} \chi_{[0,\infty)}$, then
    \begin{equation*}
        f_1 * f_2 * ... * f_n (t) = \sum_{l=1}^{n} \frac{\prod_{k=1}^{n}\lambda_k}{\prod_{k=1, k \neq l}^{n} (\lambda_k - \lambda_l)} \cdot e^{-\lambda_l t}.
    \end{equation*}
    \label{lemma2}
\end{lemma}
\begin{proof}
    Consider the Laplace transform of the convolution:
    \begin{equation}
        \mathscr{L} \left[ f_1 * f_2 * ... * f_n (t) \right] (s) = \hat{f}_1(s) \cdot ... \cdot \hat{f}_n(s) =  \prod_{k=1}^{n}\frac{\lambda_k}{s+\lambda_k},
    \end{equation}
    where $\hat{f}_i(s)$ denotes the Laplace transform of $f_i(t)$. Splitting the above product into partial fractions yields
    \begin{equation}
        \Lambda(s) := \left( \prod_{k=1}^{n} \lambda_k \right) \left( \prod_{k=1}^{n} \frac{1}{s+\lambda_k} \right) = \left( \prod_{k=1}^{n} \lambda_k \right) \left( \sum_{k=1}^{n} \frac{A_k}{s+\lambda_k} \right)
    \end{equation}
    and since all the $\lambda_k$ are distinct by assumption we get that
    \begin{equation}
        \lim_{s \rightarrow -\lambda_l} (s+\lambda_l)\Lambda(s) = \frac{\prod_{k=1}^{n}\lambda_k}{\prod_{k=1, k \neq l}^{n} (\lambda_k - \lambda_l)}.
    \end{equation}
    Thus, using the inverse transform to get back in to the time domain, we are done.
\end{proof}
\noindent We are now ready to state and prove a theorem concerning the convergence properties of $\mathbb{E}[X(t)]$ as $N \longrightarrow \infty$.
\begin{theorem}
    Let $\varepsilon \in (0,1)$ and let $V \in \mathcal{C}^1[0,1]$ be the prize function defining the Markov process $X(t)$ in (\ref{x}). Then, if we define
    \begin{equation*}
        \mathbb{E}_{1-\varepsilon} [ V(X(t)) ] := \sum_{i=1}^{\lfloor(1-\varepsilon)N\rfloor} V\left(\frac{i}{N}\right)\sum_{l = 1}^{i+1} \frac{\prod_{k = 1}^{i} \lambda_k}{\prod_{k=1, k\neq l}^{i+1}\left(  \lambda_k - \lambda_l\right)} \cdot e^{-\lambda_l t}
    \end{equation*}
    and
    \begin{equation*}
        \mathbb{E}_{\varepsilon} [ V(X(t)) ] := \sum_{i=\lceil (1-\varepsilon)N \rceil}^{N-1} V\left(\frac{i}{N}\right)\sum_{l = 1}^{i+1} \frac{\prod_{k = 1}^{i} \lambda_k}{\prod_{k=1, k\neq l}^{i+1}\left(  \lambda_k - \lambda_l\right)} \cdot e^{-\lambda_l t},
    \end{equation*}
    we have the following:
    \begin{enumerate}
        \item $\lim_{\varepsilon \rightarrow 0} \lim_{N \rightarrow \infty} \mathbb{E}_{1-\varepsilon} [ V(X(t)) ] = t - H(t-V(1)) \cdot t$
        \item $\lim_{\varepsilon \rightarrow 0} \lim_{N \rightarrow \infty} \left|\left| \mathbb{E}_{\varepsilon} [ V(X(t)) ] \right|\right|_{L^1(dt)} = 0$, \hspace{0.5cm} $t \in [0,V(1))$
        \item $\lim_{N \rightarrow \infty} p_{N}(t) = H(t - V(1))$
    \end{enumerate}
    with the limits taken in the stated order. Here $H(x)$ is the Heaviside function.
    \label{theorem1}
\end{theorem}
\begin{proof}
    We start by proving $(1)$. Using the Laplace transform and the conclusion from Lemma \ref{lemma2} we get that
    \begin{align*}
        \mathscr{L} \left[ \mathbb{E}_{1-\varepsilon} [ V(X(t)) ] \right] (s) &= \sum_{i=1}^{\lfloor(1-\varepsilon)N\rfloor} V \left( \frac{i}{N} \right) \cdot \mathscr{L} \left[ \frac{1}{\lambda_{i+1}} f_1 * ... * f_{i+1} (t) \right] (s) = \\
         &= \sum_{i=1}^{\lfloor(1-\varepsilon)N\rfloor} V \left( \frac{i}{N} \right)\frac{1}{\lambda_{i+1}} \cdot \hat{f}_1 \cdot ... \cdot \hat{f}_{i+1} (s) =\\
         &= \sum_{i=1}^{\lfloor(1-\varepsilon)N\rfloor} V \left( \frac{i}{N} \right)\frac{1}{\lambda_{i+1}} \cdot \prod_{k=1}^{i+1} \frac{\lambda_k}{s+\lambda_k},
    \end{align*}
    which is well defined for all $s\in \mathbb{C}$ such that $\mbox{Re}\{s\} > - \min_{k} \lambda_k$. Now, investigate the product in the above expression.
    \begin{align*}
        &\prod_{k=1}^{i+1} \frac{\lambda_k}{s+\lambda_k} = \text{exp} \left( \sum_{k=1}^{i+1} \text{ln} \frac{\lambda_k}{s+\lambda_k} \right) =\\
        &\text{exp} \left( \sum_{k=1}^{i+1} \text{ln}\frac{N-k+1}{(N-k)(V_{k+1} - V_k)} - \text{ln} \left( s + \frac{N-k+1}{(N-k)(V_{k+1} - V_k)} \right) \right) =\\
        &\text{exp} \left( \sum_{k=1}^{i+1} \text{ln}\frac{1-\frac{k}{N}+\frac{1}{N}}{\left(1-\frac{k}{N}\right)(V_{k+1} - V_k)} - \text{ln} \frac{ s\left( 1-\frac{k}{N} \right) \left( V_{k+1} - V_k \right) + 1 - \frac{k}{N} + \frac{1}{N} }{\left(1-\frac{k}{N}\right)(V_{k+1} - V_k)}\right) =\\
        &\text{exp} \left( \sum_{k=1}^{i+1} \text{ln} \left( 1-\frac{k}{N} + \frac{1}{N} \right) - \text{ln} \left( s\left( 1-\frac{k}{N} \right) \left( V_{k+1} - V_k \right) + \left( 1 - \frac{k}{N} + \frac{1}{N} \right) \right)  \right),
    \end{align*}
    where we take $\mbox{ln}(.)$ to be the principal branch of the complex logarithm. Using that $\text{ln}(x+1/N) = \text{ln}(x) + 1/(Nx) + \mathcal{O}(1/N^2)$ for all $x > 0$ and that $V_{k+1} - V_k = V'(\xi)/N$ for some $\xi \in (k/N, (k+1)/N)$ we get that
    \begin{eqnarray}
        \nonumber &&\text{ln} \left( 1 - \frac{k}{N} + \frac{1}{N} \right) - \text{ln} \left( s \left( 1 - \frac{k}{N} \right) (V_{k+1} - V_k) + \left( 1 - \frac{k}{N} \right) + \frac{1}{N} \right) =\\
        \nonumber &&-\text{ln}\left( s(V_{k+1} - V_k) + 1 \right) + \frac{s V'(\xi)}{N^2\left( 1-\frac{k}{N} \right) \left( s(V_{k+1}-V_k) + 1 \right)} + \mathcal{O} \left( \frac{1}{N^2} \right) =\\
        &&-\text{ln}\left( s(V_{k+1} - V_k) + 1 \right) + \mathcal{O}\left( \frac{1}{N^2} \right),
        \label{complexlog}
    \end{eqnarray}
    where we have used that $k \leq \lfloor(1-\varepsilon)N\rfloor + 1 < N$ for $N$ large enough, in the final equality. Expression (\ref{complexlog}) is local and well defined for all $s \in B_{r_N}(0)$, i.e. the disc centered at the origin with radius $r_N := \min_{k}(V_{k+1} - V_k)^{-1} \approx \min_{k} \lambda_k$. By the equality $V_{k+1} - V_k = (V'(k/N) + \mathcal{O}(1/N))/N$ and a Taylor expansion of $\text{ln} (x+1)$ about $x=0$ we get for all $s \in B_{r_N}(0)$ that
    \begin{align*}
        &\text{exp} \left( - \sum_{k=1}^{i+1} \left( \text{ln} (s(V_{k+1} - V_k) + 1) + \mathcal{O}\left( \frac{1}{N^2} \right) \right) \right)=\\
        &\text{exp} \left( - \sum_{k=1}^{i+1} \left( \frac{s}{N} \left( V'\left( \frac{k}{N} \right) + \mathcal{O}\left( \frac{1}{N} \right) \right) + s^2 \mathcal{O}\left( \frac{1}{N^2} \right) + \mathcal{O}\left( \frac{1}{N^2} \right) \right) \right) =\\
        &\text{exp} \left( -\int_{1/N}^{(i+1)/N} s \cdot V'(x) dx + \frac{i}{N} \left[ \mathcal{O}\left( \frac{1}{N} \right) + s \mathcal{O}\left( \frac{1}{N} \right) + s^2 \mathcal{O}\left( \frac{1}{N} \right) \right] \right) =\\
        &\text{exp} \left(-s\left( V\left( \frac{i+1}{N} \right) - V\left( \frac{1}{N} \right) \right) + \frac{i}{N} \cdot \Phi\left(s,\frac{1}{N}  \right) \right),
    \end{align*}
    where the ordo terms have been included in $\Phi$. Thus,
    \begin{align*}
        &\mathscr{L} \left[ \mathbb{E}_{1-\varepsilon} [ V(X(t)) ] \right] (s) = \sum_{i=1}^{\lfloor(1-\varepsilon)N\rfloor} V \left( \frac{i}{N} \right) \frac{1}{\lambda_{i+1}} \cdot e^{-s\left( V\left( \frac{i+1}{N} \right) - V\left( \frac{1}{N} \right) \right) + \frac{i}{N} \cdot \Phi\left(s,\frac{1}{N}  \right)} =\\
        &\sum_{i=1}^{\lfloor(1-\varepsilon)N\rfloor} V \left( \frac{i}{N} \right) \frac{\left(  V'\left( \frac{i}{N} \right) + \mathcal{O} \left( \frac{1}{N} \right) \right) \left( 1 + \mathcal{O}\left( \frac{1}{N} \right) \right)}{N} \cdot e^{-s\left( V\left( \frac{i+1}{N} \right) - V\left( \frac{1}{N} \right) \right) + \frac{i}{N} \cdot \Phi\left(s,\frac{1}{N}  \right)} =\\
        &\int_{1/N}^{1-\varepsilon} V(x)\left( V'(x) + \mathcal{O}\left( \frac{1}{N} \right)\right) \left( 1 + \mathcal{O}\left( \frac{1}{N} \right) \right) e^{-s\left( V(x) - V(1/N) \right) + \frac{i}{N} \cdot \Phi\left(s,\frac{1}{N}  \right)} dx,
    \end{align*}
    and by considering the $N$-player limit of this expression we finally get (making the change of variables $du = V'(x)dx$ using that $V(x)$ is increasing) that
    \begin{align*}
        \lim_{N \rightarrow \infty}\mathscr{L} \left[ \mathbb{E}_{1-\varepsilon} [ V(X(t)) ] \right] (s) &= \int_0^{1-\varepsilon} V(x)V'(x)e^{-sV(x)}dx = \int_0^{V(1-\varepsilon)} ue^{-su} du =\\
        &= \frac{1}{s^2} - \left( \frac{V(1-\varepsilon)}{s}  + \frac{1}{s^2} \right) e^{-sV(1-\varepsilon)},
    \end{align*}
    which in turn yields
    \begin{equation*}
        \lim_{\varepsilon \rightarrow 0} \lim_{N \rightarrow \infty} \mathscr{L} \left[ \mathbb{E}_{1-\varepsilon} [ V(X(t)) ] \right] (s) = \frac{1}{s^2} - \left( \frac{V(1-\varepsilon )}{s} + \frac{1}{s^2} \right) e^{-sV}.
    \end{equation*}
    Since $V(x)$ is bounded on a compact interval we can use Lebesgue's theorem on dominated convergence to interchange the order between the limits and the Laplace transform and then, by the inversion formula, we get that
    \begin{equation*}
        \lim_{\varepsilon \rightarrow 0} \lim_{N \rightarrow \infty} \mathbb{E}_{1-\varepsilon} [ V(X(t)) ] = t - H(t-V(1)) \cdot t
    \end{equation*}
    For proving $(2)$ it will by assuming $(3)$ suffice to consider the $L^1$-norm of the sum from $\lceil (1-\varepsilon)N \rceil$ to $N-2$.
    \begin{eqnarray*}
        &&\left|\left| \sum_{i = \lceil (1-\varepsilon)N \rceil }^{N-2}V\left( \frac{i}{N} \right) p_{i+1} \right|\right| = \left|\left| \sum_{i = \lceil (1-\varepsilon)N \rceil }^{N-2}V\left( \frac{i}{N} \right) \cdot \frac{1}{\lambda_{i+1}} \left( \lambda_1 e^{-\lambda_1 t} * ... * \lambda_{i+1} e^{-\lambda_{i+1} t} \right) \right|\right| \leq\\
        &&\leq \sum_{i = \lceil (1-\varepsilon)N \rceil }^{N-2}V\left( \frac{i}{N} \right) \cdot \frac{1}{\lambda_{i+1}} \leq \sum_{i = \lceil (1-\varepsilon)N \rceil }^{N-2}V\left( \frac{i}{N} \right) \cdot \frac{V_{i+2} - V_{i+1}}{1/N}\cdot \frac{1}{N} =\\
        &&= \sum_{i = \lceil (1-\varepsilon)N \rceil }^{N-2} V\left( \frac{i}{N} \right) \left( V'\left( \frac{i}{N} \right) + \mathcal{O}\left( \frac{1}{N} \right) \right)\cdot \frac{1}{N} =\\
        &&= \int_{1-\varepsilon}^1 V(x)\left( V'(x) + \mathcal{O}\left( \frac{1}{N} \right) \right)dx + \mathcal{O}\left( \frac{1}{N} \right) \leq\\
        &&\leq \varepsilon \cdot \text{sup} \left| V(x)V'(x) \right| + \mathcal{O}\left( \frac{1}{N} \right),
    \end{eqnarray*}
    where we used the triangle inequality and the fact that the convolution is a probability density in the second inequality. The statement in $(2)$ follows immediately. Finally, for proving the statement in $(3)$ we analyse the limiting behavior of $p'_N(t)$. Note that
    \begin{equation*}
        p_N(t) = \sum_{l=1}^{N-1} \frac{\prod_{k=1}^{N-1}\lambda_k}{\prod_{k=1, k \neq l}^{N} (\lambda_k - \lambda_l)} \cdot e^{-\lambda_l t} + 1.
    \end{equation*}
    By Lemma \ref{lemma2} we get
    \begin{eqnarray*}
        \frac{d}{dt} p_N(t) = \sum_{l=1}^{N-1} \frac{\prod_{k=1}^{N-1}\lambda_k}{\prod_{k=1, k \neq l}^{N-1} (\lambda_k - \lambda_l)} \cdot e^{-\lambda_l t} = \left(\lambda_1 e^{-\lambda_1 t} * ... * \lambda_{N-1} e^{-\lambda_{N-1} t} \right)(t).
    \end{eqnarray*}
    From the proof of $(1)$ we know that
    \begin{equation*}
        \mathscr{L} \left[ \frac{d}{dt}p_N(t) \right] = \prod_{k=1}^{N-1} \frac{\lambda_k}{\lambda_k + s} = e^{-s \left( V\left( \frac{N-1}{N} \right) - V\left( \frac{1}{N} \right) \right) + \frac{N-1}{N}\cdot \Phi \left( s, \frac{1}{N} \right)},
    \end{equation*}
    and it follows that $\lim_{N\rightarrow \infty} p'_N(t) = \delta_{V(1)}(t)$. This proves $(3)$ since $p_N(0) = 0$ and $\lim_{t \rightarrow \infty} p_N(t) = 1$.
\end{proof}
\noindent Following the same line of reasoning as in the proof of Theorem \ref{theorem1} one can also prove that $\lim_{N \rightarrow \infty}\text{Var}(V(X(t))) := \lim_{N \rightarrow \infty}(\mathbb{E}[V(X(t))^2] - \mathbb{E}[V(X(t))]^2) = 0$ on the interval $t \in [0,V(1))$. We collect our results in a corollary.
\begin{corollary}
    Let $V(x)$ be an increasing $\mathcal{C}^1$-function defined on the unit interval and let $X(t)$ be defined as in (\ref{x}). Then
    \begin{enumerate}[(i)]
        \item $\lim_{N \rightarrow \infty} \mathbb{E}[X(t)] = V^{-1}(t) := q(t), \hspace{0.5cm} t \in [0,V(1))$
        \item $\lim_{N \rightarrow \infty}$ $\mathrm{Var}(X(t)) = 0, \hspace{0.5cm} t \in [0,V(1)) $
    \end{enumerate}
\end{corollary}
\noindent Below we have included some numerical results illustrating the convergence of $\mathbb{E}[X(t)]$ and $\mbox{Var}(X(t))$ for two different choices of $V(x)$.
\begin{figure}[H]
    \begin{center}
        \includegraphics[scale=0.34]{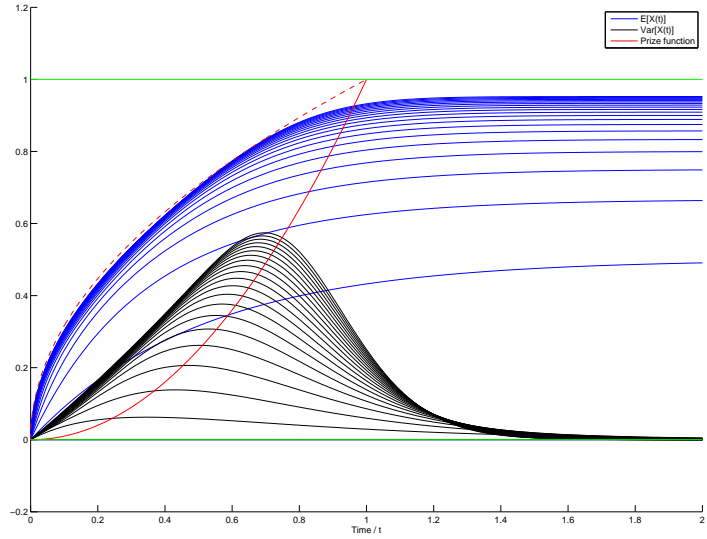}
        \caption{Convergence of $\mathbb{E}[X(t)]$ and $\mbox{Var}(X(t))$ when $V(x) = x^2$.}
    \end{center}
\end{figure}
\begin{figure}[H]
    \begin{center}
        \includegraphics[scale=0.34]{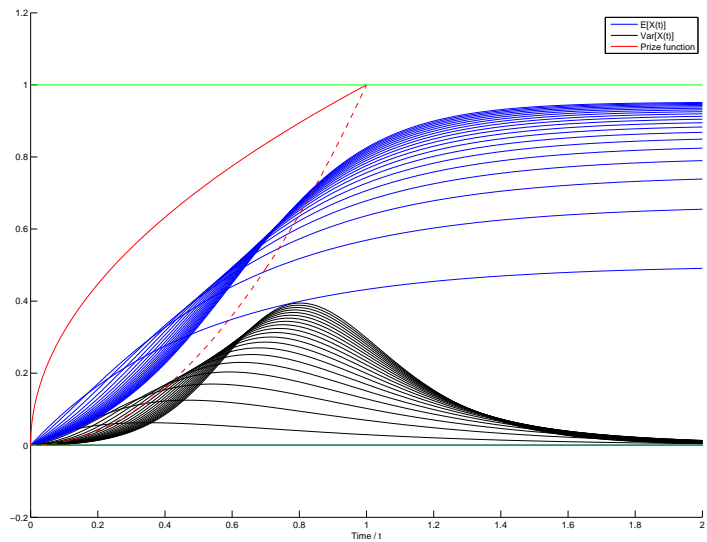}
        \caption{Convergence of $\mathbb{E}[X(t)]$ and $\mbox{Var}(X(t))$ when $V(x) = \sqrt{x}$.}
    \end{center}
\end{figure}
\noindent The $N$-player limit in the dynamic generalization of the WA introduces some new features in the game. In contrast to the case of having a finite number of players, in the limit, there will be a continuous flow of players quitting the game. This flow depends on the behavior of the mixed strategies that the players use to pick waiting times after a certain fraction of the players have left at time $t$. It is natural to believe that the only characteristic of the strategies that determines the flow of players is the behavior near $\tau = 0$. Since there are infinitely many players (all using the same strategy) there will always be players having waiting times arbitrarily close to $\tau = 0$. In what follows we will analyse the connection between the out-flow of players and the characteristics of the mixed strategies.\\
\indent Suppose that $\{ f_i^{\alpha} \}_{i=1}^N$ and $\{ f_i^{\beta} \}_{i=1}^N$ are two different sequences of probability densities in $\mathcal{C}^1(\mathbb{R}_+)$ such that $f_i^{\alpha}(0) = f_i^{\beta}(0) \neq 0$ for all $i = 1, ..., N$. We think of $f_i^{\alpha}$ or $f_i^{\beta}$ as different choices of strategies used in the $i$'th round of an $N$-player dynamic model WA. Let $\Omega$ represent both $\alpha$ and $\beta$ and consider an independent sample $\tau_1^{\Omega}, ..., \tau_{N-i+1}^{\Omega} \sim f_i^{\Omega}d\tau$. We define
\begin{equation}
    T_i^{\Omega} := \text{min}\left( \tau_1^{\Omega}, ..., \tau_{N-i+1}^{\Omega} \right).
\end{equation}
Note that
\begin{align*}
    \mathbb{P}\left\{ T_i^{\Omega} \leq x \right\} &= 1 - \mathbb{P}\left\{ \tau_1^{\Omega} \geq x, ..., \tau_{N-i+1}^{\Omega} \geq x \right\} =\\
     &= 1 - \prod_{k=1}^{N-i+1}\mathbb{P}\left\{ \tau_k^{\Omega} \geq x \right\} = 1 - \left( 1 - F_i^{\Omega}(x) \right)^{N-i+1},
\end{align*}
where $F_i^{\Omega}$ is the cdf associated to the density function $f_i^{\Omega}$, i.e. $F_i^{\Omega} (x) := \mathbb{P}\left\{ \tau_k^{\Omega} \leq x \right\}$. Thus, the density function of $T_i^{\Omega}$ is given by
\begin{equation}
    f_{(i)}^{\Omega}(\tau) = (N-i+1)f_i^{\Omega}(\tau)\left( 1 - F_i^{\Omega}(\tau) \right)^{N-i}.
\end{equation}
Note that
\begin{eqnarray}
    \lim_{N\rightarrow \infty} \int_0^{\infty} f_{(i)}^{\Omega}(\tau) \cdot \varphi(\tau) d\tau = \varphi(0) + \lim_{N\rightarrow \infty}\int_0^{\infty} f_{(i)}^{\Omega}(\tau) (\varphi(\tau) - \varphi(0)) d\tau = \varphi(0)
    \label{dist}
\end{eqnarray}
for any test function $\varphi \in \mathcal{C}^{\infty}_0 \left( \mathbb{R} \right)$, so $\lim_{N\rightarrow \infty}f_{(i)}^{\Omega} d\tau = \delta_0 $. Finally we also define the sum
\begin{equation}
    S_{q}^{\Omega} := \sum_{i=1}^{\lfloor qN \rfloor} T_i^{\Omega}, \hspace{0.5cm} q \in (0,1].
\end{equation}
Now, let $\delta > 0$ and consider the probability
\begin{eqnarray}
    \nonumber &&\mathbb{P} \left\{ \left| S_{ q }^{\alpha} - S_{ q }^{\beta} \right| \geq \delta \right\} \leq \frac{1}{\delta ^2} \mathbb{E} \left[ \left( S_{ q }^{\alpha} - S_{ q }^{\beta} \right)^2 \right] =\\
    &&= \frac{1}{\delta ^2} \left( \mathbb{E}\left[ \left( S_{ q }^{\alpha} \right)^2 \right] + \mathbb{E}\left[ \left( S_{ q}^{\beta} \right)^2 \right] - 2\mathbb{E}\left[ S_{ q }^{\alpha} \right]\mathbb{E}\left[ S_{ q }^{\beta} \right]\right) =\\
    \nonumber &&= \frac{1}{\delta ^2}\left( \text{Var}\left( S_{ q }^{\alpha} \right) + \text{Var}\left( S_{ q }^{\beta} \right) + \mathbb{E}\left[ S_{ q }^{\alpha} \right]^2 + \mathbb{E}\left[ S_{ q }^{\beta} \right]^2 - 2\mathbb{E}\left[ S_{ q }^{\alpha} \right]\mathbb{E}\left[ S_{ q }^{\beta} \right] \right),
    \label{chebyshev}
\end{eqnarray}
where we used the Chebyshev-Markov inequality in the first step. Our goal is to prove that the right most expression in the inequality above, under reasonable assumptions on the functions $\left\{ f_i^{\Omega} \right\}_{i=1}^N$, tend point wise to zero as $N$ tends to infinity. If the random sums $S_q^{\Omega}$ are converging a.s.-$d\mathbb{P}$ then, in the limit, the variances vanish and the problem would reduce to proving that $\lim_{N \rightarrow \infty} \mathbb{E}\left[ S_{ q }^{\alpha} - S_{ q }^{\beta} \right] = 0$ for all $q \in [0,1)$. We start by investigating the convergence of the sums. According to \cite{Durett} (theorem 8.3, pp. 62) it is sufficient to prove that $\sum_{i=1}^{\infty} \text{Var}(T_i^{\Omega}) < \infty$ in order to get almost sure convergence in $S_q^{\Omega}$. This property is proven in the following lemma:
\begin{lemma}
    Assume that $\left\{ f_i^{\Omega} \right\}_{i=1}^{N} \subset \mathcal{C}^1(\mathbb{R}_+)$ is uniformly bounded from above by some positive constant $C$, independent of $N$, and that the $f_i^{\Omega}(\tau)$ decay at least like $o(1/\tau^p)$, with $p > 3$. Then $\sum_{i=1}^{\infty} \mathrm{Var}(T_i^{\Omega}) < \infty$.
    \label{lemma3}
\end{lemma}
\begin{proof}
    By partial integration and using the fact that $f_i^{\Omega}(\tau) = o\left( 1/\tau^{3} \right)$, we have
    \begin{align*}
        \mathbb{E} \left[ \left( T_i^{\Omega} \right)^2 \right] &= \int_0^{\infty} \tau^2 \cdot (N-i+1)f_i^{\Omega}(\tau)\left( 1 - F_i^{\Omega}(\tau) \right)^{N-i} d\tau =\\
        &= 2\int_0^{\infty} \tau \cdot \left( 1 - F_i^{\Omega} (\tau) \right)^{N-i+1} d\tau =\\
        &= 2 \int_0^{\Delta} \tau \cdot \left( 1 - F_i^{\Omega} (\tau) \right)^{N-i+1} d\tau + 2 \int_{\Delta}^{\infty} \tau \cdot \left( 1 - F_i^{\Omega} (\tau) \right)^{N-i+1} d\tau,
    \end{align*}
    where $\Delta$ is an arbitrary positive number. Because of (\ref{dist}) it is clear that the second moment of $T_i^{\Omega}$ tends to zero as $N$ grows and the result can be established by proving that this convergence is sufficiently fast. We investigate the rate of convergence for both of the integrals above. Let $\bar{\tau} = \text{argmax}_{[0,\Delta]} \{ \tau \cdot \left( 1 - F_i^{\Omega} (\tau) \right)^{N-i+1} \}$. Then
    \begin{eqnarray*}
        \int_0^{\Delta} \tau \cdot \left( 1 - F_i^{\Omega} (\tau) \right)^{N-i+1} d\tau \leq \bar{\tau} \left( 1 - F_i^{\Omega} (\bar{\tau}) \right)^{N-i+1} \cdot \Delta \leq \Delta^2 e^{(N-i+1)\mbox{ln}(1-F_i^{\Omega} (\bar{\tau}))},
    \end{eqnarray*}
    and we get exponential convergence near the origin since $\bar{\tau} > 0$ and hence also $F_i^{\Omega} (\bar{\tau})) > 0$ since $f_i^{\Omega}(0) > 0$ by assumption. For the other part, let $k \in \mathbb{N}$ be the least integer so that the integral of the function $\tau \cdot \left( 1 - F_i^{\Omega} (\tau) \right)^k$ converges at infinity. By the asymptotic assumption on $f_i^{\Omega}(\tau)$ an easy calculation shows that $k=1$. Therefore
    \begin{align*}
        \int_{\Delta}^{\infty} \tau \cdot \left( 1 - F_i^{\Omega} (\tau) \right)^{N-i+1} d\tau &\leq \left( 1 - F_i^{\Omega} (\Delta) \right)^{N-i} \int_{\Delta}^{\infty} \tau \cdot \left( 1 - F_i^{\Omega} (\tau) \right) d\tau \leq \\
        & \leq \Delta^2 e^{(N-i)\mbox{ln}(1-F_i^{\Omega} (\Delta))},
    \end{align*}
    and once again we get exponential convergence. Thus, because of the uniform bound of $\left\{ f_i^{\Omega} \right\}_{i=1}^{N}$, we are done.
\end{proof}
\noindent Next we consider the limiting properties of the expectation $\mathbb{E}\left[ S_{ q }^{\alpha} - S_{ q }^{\beta} \right]$. Note that
\begin{equation*}
    \mathbb{E}\left[ S_{ q }^{\alpha} - S_{ q }^{\beta} \right] = \sum_{i=1}^{\lfloor qN \rfloor} \mathbb{E} \left[ T_i^{\alpha} - T_i^{\beta} \right].
\end{equation*}
As with the variances in the previous lemma we are interested in the rate of convergence to zero of $\mathbb{E} \left[ T_i^{\alpha} - T_i^{\beta} \right]$ as $N$ tends to infinity. We summarize the results in a lemma:
\begin{lemma}
    Assume that $\left\{ f_i^{\Omega} \right\}_{i=1}^{N} \subset \mathcal{C}^1(\mathbb{R}_+)$ is uniformly bounded from above by some positive constant $C$, independent of $N$, and that the $f_i^{\Omega}(\tau)$ decay at least like $o(1/\tau^p)$, with $p > 1$, when $\tau \rightarrow \infty$. Then $\lim_{N \rightarrow \infty}\mathbb{E}\left[ S_{ q }^{\alpha} - S_{ q }^{\beta} \right] = 0$.
    \label{lemma4}
\end{lemma}
\noindent The proof of this lemma goes exactly like the proof of Lemma \ref{lemma3}. Now, using the Chebyshev estimate in (\ref{chebyshev}) together with the results in Lemma \ref{lemma3} and \ref{lemma4} we have reached the following conclusion:
\begin{theorem}
    Let $\left\{ f_i^{\Omega} \right\}_{i=1}^{N} \subset \mathcal{C}^1(\mathbb{R}_+)$ be uniformly bounded from above by some positive constant $C$, independent of $N$, and assume that $f_i^{\Omega}(\tau)$ decay at least like $o(1/\tau^p)$, with $p > 3$. Then
    \begin{equation*}
        \lim_{N \rightarrow \infty}\mathbb{P} \left\{ \left| S_{ q }^{\alpha} - S_{ q }^{\beta} \right| \geq \delta \right\} = 0
    \end{equation*}
    for all $\delta > 0$ and all $q \in (0,1]$.
    \label{theorem2}
\end{theorem}
\noindent The intuitive meaning of Theorem \ref{theorem2} is that the time evolution of an $N$-player dynamic model WA with a large $N$ is completely determined by the initial values of the mixed strategies that are being used. In terms of the mean field density $m(t,\tau)$ from the introduction this means the $t$-marginal $m(t,0)$. Thus, in the $N$-player limit the $\tau$-dependence, for $\tau > 0$, gets superfluous and it suffices to choose a strategy in $t$. In a sense this makes the dynamic model look like a static model since it is enough for each player to decide upon "how to quit". In \cite{HaigCannings}, the static model is the $N$-player version of the WA differing from the dynamic model in that the remaining players are not allowed to reconsider their waiting times as other players drop out. This model has not yet been thoroughly investigated and the analysis done so far in the $N$-player case neither give explicit results, nor does it give any concrete answers for when and if there exist evolutionary stable strategies. Due to Theorem \ref{theorem2} there might be reason to believe that the two models coincide in the $N$-player limit. In what follows we will investigate this connection in order to find more definite results for the static model.

\section{Convergence in the static model}
In the static model of the $N$-player WA the players begin by choosing their waiting times for the upcoming game. Once the game has started the participants are not allowed to reappraise their bids and (after a reordering of the players) we get an increasing sequence of waiting times $t_{1}, t_{2}, ..., t_{N}$. The prizes $\{ V_k \}_{k=1}^N \subset \mathbb{R}_+$ are handed out to each player according to this order, i.e. the player with the least waiting time receives $V_1$ and pays $t_{1}$, the player with the second least receives $V_2$ and pays $t_{2}$, and so forth until the last player who receives $V_N$ and pays $t_{N-1}$. Thus, the game ends once the second last player quits. In contrast to the dynamic model, which was a repeated game, the static model of WA is a one-shot game. We are interested in finding a probability density $g_N(t)$ being a mixed strategy ESS in this $N$-player model. In particular, since every ESS also is a Nash-equilibrium, such a strategy must have the property that if one player chooses to play a pure strategy $\delta_x$, $x \in \mathbb{R}_+$, and the rest of the players are playing according to $g_N(t)$, the expected payoff of playing the pure strategy must be constant regardless of the value of $x$. For the expected payoff we adopt the following notation:
\begin{definition}
\label{definition2}
Let $\{ \mu_k \}_{k=1}^N$ be a sequence of probability measures on $\mathbb{R}_+$ representing a choice of mixed strategies by the players in an $N$-player static model WA and let $\mu_{-i}~:=~\{ \mu_k \}_{k=1, k\neq i}^N$ for all $i = 1, ..., N$. The expected payoff to player $i$ when playing $\mu_i$ against $\mu_{-i}$ is denoted $\mathcal{J}_N(\mu_i | \mu_{-i}) = \mathcal{J}_N(\mu_i | \mu_1, ..., \mu_{i-1}, \mu_{i+1}, ..., \mu_N)$. If a given number $r$, $2 \leq r \leq N-1$, of the measures in $\mu_{-i}$ are equal to the same measure $\mu$ we write $(\mu, ..., \mu) = \mu^{\oplus r}$\footnote{In this text the measures in $\{ \mu_k \}_{k=1}^N$ will often be given by a density function $\varphi_k(t)$, i.e. $\mu_k = \varphi_k(t)dt$, and we will in those cases abuse the notation in Definition \ref{definition2} by identifying $\mu_k$ with $\varphi_k$.}.
\end{definition}
Playing $\delta_x$ in a population where all opponents play $g_N(t)$, and $G_N(t)$ is the cdf of $g_N(t)$, we have that
\begin{align}
    \nonumber \mathcal{J}_N\left(\delta_x| g_N^{\oplus(N-1)}\right) =& \sum_{r=0}^{N-2} \left( V_{r+1} - x \right) \binom{N-1}{r}G_N(x)^r\left( 1 - G_N(x) \right)^{N-r} +\\
    &+ \int_0^x \left( V_N - y \right) d\left( G_N(y)^{N-1} \right).
\end{align}
If we require that $d/dx \left[ \mathcal{J}_N\left(\delta_x| g_N^{\oplus(N-1)}\right) \right]= 0$ we end up with the necessary condition for $G_N$ to be an ESS:
\begin{equation}
\left\{
    \begin{array}{l}
    \frac{dG_N}{dx} = \frac{1 - G_N^{N-1}}{(N-1)\sum_{r=0}^{N-2} c_{r} \binom{N-2}{r} G_N^r\left( 1-G_N \right)^{N-2-r}} =: \Xi_N(G_N)\\
     G_N(0) = 0,
    \end{array}
\right.
\label{autonomous}
\end{equation}
where $c_{r} := \left( V_{r+2} - V_{r+1} \right)$. As in the previous section we will from now on assume the prize sequence $\{ V_k \}_{k=1}^N$ to be positive and strictly increasing. This means that $c_{r} > 0$ for all $r = 0, ..., N-2$ and hence the function $\Xi_N(\xi)$ is well defined, positive and continuous for all $0 \leq \xi \leq 1$. We also note that $\Xi_N(1) = 0$ and $\Xi_N'(1) < 0$, i.e. (\ref{autonomous}) is asymptotically stable in $\xi = 1$, and since $dG_N/dx(0) > 0$ we get by basic properties of autonomous equations (see e.g. \cite{Olver}) that the unique solution to (\ref{autonomous}), given any increasing prize sequence, is the cdf of some probability density function $g_N(t)$. In the following proposition we prove that the limiting solution of the static model coincides with the limiting solution of the dynamic model.
\begin{proposition}
Let $V(x)$ be an increasing $\mathcal{C}^1$-function on the unit interval such that $V(0)=0$ and define the sequence $\left\{ V_k \right\}_{k=1}^N$ by $V_k := V(k/N)$. Then, if $G_N$ is the unique solution of problem (\ref{autonomous}) with the given prize sequence, it holds that
\begin{equation*}
    G_N(t) \longrightarrow
    \left\{
        \begin{array}{l}
            V^{-1}\left( t \right), \hspace{0.5cm} 0 \leq t \leq V(1)\\
            1, \hspace{0.5cm} t > V(1)
        \end{array}
    \right.
\end{equation*}
uniformly as $N \rightarrow \infty$.
\label{proposition1}
\end{proposition}

\begin{proof}
The key idea to prove this proposition is to consider the limiting equation of (\ref{autonomous}) as $N \rightarrow \infty$. For this we use a theorem by Bernstein (see e.g. \cite{Lorentz}) saying that if $f: [0,1] \rightarrow \mathbb{R}$ is a bounded continuous function and $p_{n\nu}(x)$, $\nu = 0, ..., n$, is the $\nu$'th Bernstein polynomial of degree $n$ then
\begin{equation*}
    \lim_{n \rightarrow \infty} \sum_{\nu=0}^{n} f\left( \frac{\nu}{n} \right) p_{n\nu}(x) = f(x),
\end{equation*}
and the relation holds uniformly on $[0,1]$. Thus, considering the denominator in (\ref{autonomous}) as a function on $[0,1]$, replacing $G_N$ by a fixed $x$, we get uniform convergence to $V'$ since $(N-1)c_r = V'(r/N) + \mathcal{O}(1/N)$. Consequently
\begin{equation*}
    \lim_{N \rightarrow \infty} \Xi_N(x) = \frac{1}{V'(x)},
\end{equation*}
uniformly on the interval $[0,1)$. We therefore have the limit equation $y' = 1/V'(y)$ satisfying the initial value $y(0)=0$, which admits the unique solution $y(x) = V^{-1}(x)$ on $[0,V(1)]$. Let $y_N$ be the unique solution to (\ref{autonomous}) restricted to the interval $[0,V(1)]$ and consider the absolute value of the difference:
\begin{align*}
    |y_N(x) - y(x)| &= \left| \int_0^x \Xi_N(y_N(t)) - \frac{1}{V'(y(t))} dt  \right| \\
    & \leq \int_0^x \left| \Xi_N(y_N(t)) - \Xi_N(y(t)) \right| dt + \int_0^x \left| \Xi_N(y(t)) - \frac{1}{V'(y(t))} \right| dt \\
    & \leq C_N \int_0^x \left| y_N(t) - y(t) \right| dt + \varepsilon_N x,
\end{align*}
where $C_N$ is a Lipschitz constant (uniformly bounded over $N$) of $\Xi_N$ and $\lim_{N \rightarrow \infty}\varepsilon_N =~0$ by the uniform convergence of $\Xi_N$. Thus, by a Gr\"{o}nwall estimate we get the point wise upper bound
\begin{equation}
    |y_N(x) - y(x)| \leq \varepsilon_N x e^{xC_N}, \hspace{0.5cm} x \in [0,V(1)]
    \label{Gronwall}
\end{equation}
and hence point wise convergence of $y_N$ in this interval. Point wise convergence on a compact interval does not in general imply uniform convergence, but since $\{ y_N \}_{N=2}^{\infty}$ is a sequence of monotone functions that converge point wise to a continuous function we have by a theorem in \cite{BuchananHildebrandt} in this case even uniform convergence. For the rest of the half axis, i.e. $x \in (V(1), \infty)$, we have uniform convergence towards 1 by monotonicity and the fact that $y(V(1)) = V^{-1}(V(1)) = 1$. This finally proves the theorem.
\end{proof}
\noindent Proposition \ref{proposition1} proves that the limiting properties of the static and the dynamic models coincide. In the $N$-player limit of the static model we have concluded that, given an increasing prize function $V(x)$ an ESS, if it exists, must be the strategy of choosing a waiting time according to the probability density $\dot{q}(t)$. On the other hand, we also know that the same strategy is reached as the limit of ESS strategies in the dynamic $N$-player model. Therefore, since Theorem \ref{theorem2} suggests that the dynamic model is indistinguishable from the static model in the $N$-player limit, one might have hope for the strategy $\dot{q}(t)$ to be an ESS even in the static model. A more thorough analysis of the limiting strategy, carried out in Section \ref{section5},  shows that this is false. The results will nevertheless give hints on how to proceed with the $N$-player static model.\\
\indent Another notable result from this section is that the candidate ESS-solution $g_N(t)$ has support in all of $\mathbb{R}_+$ (due to asymptotic stability) despite the fact that the prizes are bounded from above by $V_N$. This property might feel somewhat unnatural since the mass of $g_N(t)$ in $(V_N,\infty)$ would contribute negatively to the expected payoff. This is not the case, however, since the last player quitting is paying the time cost of the second last player. The tail of $g_N(t)$ is important for the strategy to be an ESS, and even to be a Nash-equilibrium. Indeed, if we consider an $N$-player game of the static WA in which all players use a strategy $\varphi(t)dt$ such that $\varphi \in \mathcal{C}[0,V_N]$ one could choose to play $\delta_{V_N}$ which would have an expected payoff strictly larger than that of playing $\varphi(t)dt$. A good way to think of the asymptotic behaviour of $g_N(t)$ (for large values of $t$) is that the mass in $(V_N,\infty)$ should be small enough for the probability of having more that one out of $N$ trails in this interval to be negligibly small. Hence if an unlucky player receives a waiting time much larger than $V_N$ he will most likely be "saved" by the second last player. By this one would expect the tail to become lighter and lighter as the number of players increase since the probability of getting several players in $(V_N,\infty)$ otherwise would increase. This property is partly supported by the conclusion of Proposition \ref{proposition1} saying that the limit strategy is confined to the interval $[0,V_N]$.

\section{Properties of the limit strategy and the $N$-player static model}
\label{section5}
In this section we investigate the game theoretic properties of the limiting strategy $\dot{q}(t)$. In particular we are interested in knowing wether this strategy represents an ESS or not. The first problem one encounters when initiating this analysis is that the definition of a mixed strategy ESS does not make sense in the $N$-player limit since a finite number of invading players do not affect an infinite population. In order to give a more suitable definition we consult the abstract, but still standard, measure theoretic approach to games and game theory. Such an approach is given e.g. by Balder \cite{Balder}, and goes as follows:\\
\indent Let $\mathfrak{P} = (P, \mathcal{P},\mu)$ be a finite measure space such that, for convenience, $\mu(P) = 1$. Each $p \in P$ can be thought of as a \textit{player} and hence $\mathfrak{P}$ is the space of players. Let $A$ be a metric space of \textit{actions}, or pure strategies, available to the players in $P$\footnote{In \cite{Balder} the setting is even more general with each player $p \in P$ having access to some action set $A_p \in \mathcal{B}(A).$}, let $\mathcal{B}(A)$ be the Borel $\sigma$-algebra, and consider the pair $\mathfrak{A} = (A, \mathcal{B}(A))$. The set of all probability measures on $\mathfrak{A}$ is denoted by $\mathcal{M}_1(A)$. A \textit{mixed (action) profile} is a Young measure $\Delta : P \rightarrow \mathcal{M}_1(A)$ such that for $\mu$-a.e. $p$ it holds that $\Delta(p)(A) = 1$. Recall that for $\Delta$ to be a Young measure we require the map $p \mapsto \Delta(p)(B)$ to be $\mathcal{P}$-measurable for all fixed $B \in \mathcal{B}(A)$ . The set of all such $\Delta$'s is denoted $\mathcal{R}$. Intuitively a given $\Delta \in \mathcal{R}$ contains information of what mixed strategies the players in $P$ have chosen. If $f$ is a measurable selection of the multifunction $\Sigma : p \mapsto A$ and we choose $\epsilon_f \in \mathcal{R}$ so that $\epsilon_f (p) := \delta_{f(p)}$ (Dirac measure at $f(p)$) we see that the \textit{pure} action profiles are contained in $\mathcal{R}$. Finally we define the \textit{payoff function} connected to player $p \in P$ as a function $\mathcal{J}_p: P \times \mathcal{R} \rightarrow [-\infty, \infty)$. Thus, $\mathcal{J}_p(p | \Delta)$ measures the benefits (or losses) of player $p$ when the population is playing $\Delta \in \mathcal{R}$. For our purpose it suffices to consider games in which all players in $P$ share the same explicit form of the payoff function so that $\mathcal{J}_p(p|.) = \mathcal{J}(p|.)$ for all $p \in P$. In this setting we define a game as a triplet $\mathfrak{G} = (\mathfrak{P}, \mathfrak{A}, \mathcal{J})$.\\
\indent Given a mixed action profile $\Delta \in \mathcal{R}$ and a player $p \in P$ it is natural to think of $\Delta$ as being made up of two parts. One "internal" part representing the strategy chosen by player $p$ (namely $\Delta(p)$) and one "external" part for the strategies of the opponents of $p$. In many game theoretic models (see \cite{Balder}) the internal and external parts of $\Delta$ are reflected in an explicit form of the payoff function called \textit{internal-external form}. For this we need (i) a space $Y$ that we call the space of \textit{profile statistics} of the game, (ii) a \textit{utility} function $\mathcal{U}:A \times Y \rightarrow [-\infty,\infty)$ and (iii) a mapping $e:\mathcal{R} \rightarrow Y$, called the \textit{mixed externality}, such that $\mathcal{J}$ can be written as:
\begin{equation*}
    \mathcal{J}(p | \Delta) = \int_A \mathcal{U}(x,e(\Delta))\Delta(p)(dx), \hspace{1cm} p \in P
\end{equation*}
With the abstract framework at hand we can now extend Balder's ideas to give a measure theoretical notion of an ESS. Before doing so we recall the definition of the $N$-player ESS according to Palm \cite{Palm}.
\begin{definition}
A probability measure $\pi_*$ is said to be an \textit{evolutionary stable strategy}, or \textit{ESS}, if either
\begin{equation*}
    \text{(i)} \hspace{0.5cm} J_N\left(\pi_* | \pi_*^{\oplus(N-1)}\right) > J_N\left(\varphi | \pi_*^{\oplus(N-1)}\right) \hspace{0.6cm}
\end{equation*}
for any other probability measure $\varphi \neq \pi_*$, or else, if there is a $\bar{\varphi}$ such that equality holds in (i), then
\begin{equation*}
    \text{(ii)} \hspace{0.5cm} J_N\left(\pi_* | \pi_*^{\oplus(N-2)}, \bar{\varphi} \right) > J_N\left(\bar{\varphi} | \pi_*^{\oplus(N-2)}, \bar{\varphi}\right).
\end{equation*}
\end{definition}
\noindent For a continuum of players we extend the above definition to the following:
\begin{definition}
Let $\mathfrak{G} = (\mathfrak{P}, \mathfrak{A}, \mathcal{J})$ be a game admitting a decomposition of $\mathcal{J}$ to internal-external form with space of profile statistics $Y$, utility function $\mathcal{U}$ and mixed externality $e$. Let $\varepsilon > 0$ and pick a Borel set $P_{\varepsilon} \in \mathcal{P}$ such that $0 < \mu(P_{\varepsilon}) \leq \varepsilon$ and define the mixed action profile
\begin{equation}
    \Pi_{(\pi, \varphi)}^{P_{\varepsilon}}(p) :=
     \left\{
        \begin{array}{l}
            \pi, \text{if } p \in P \backslash P_{\varepsilon}\\
            \varphi, \text{if } p \in P_{\varepsilon}
        \end{array}
    \right.
\end{equation}
where $\pi, \varphi \in \mathcal{M}_1(A)$. We say that $\pi_* \in \mathcal{M}_1(A)$ is an ESS of $\mathfrak{G}$ if either of the following holds for all $\varepsilon$ small enough and independently of $P_{\varepsilon}$:
\begin{equation*}
   \text{(i)} \hspace{0.5cm} \int_A \mathcal{U}(x,e(\Pi_{(\pi_*, \varphi)}^{P_{\varepsilon}}))\pi_*(dx) > \int_A \mathcal{U}(x,e(\Pi_{(\pi_*, \varphi)}^{P_{\varepsilon}}))\varphi(dx),
\end{equation*}
for any $\varphi \in M_1^+(A) \backslash \{ \pi_* \}$ and all $p \in P \backslash P_{\varepsilon}$, or else, if there is a $\bar{\varphi} \in  \mathcal{M}_1(A)\backslash \{ \pi_* \}$ such that equality holds in (i), then
\begin{equation*}
   \text{(ii)} \hspace{0.5cm} \int_A \mathcal{U}(x,e(\Pi_{(\bar{\varphi}, \pi_*)}^{P_{\varepsilon}}))\pi_*(dx) > \int_A \mathcal{U}(x,e(\Pi_{(\bar{\varphi}, \pi_*)}^{P_{\varepsilon}}))\bar{\varphi}(dx).
\end{equation*}
for all $p \in P_{\varepsilon}$.
\label{definition1}
\end{definition}
\noindent The intuition of the above definition is the same as in the $N$-player case, namely that; playing the strategy $\pi$ in a population where all but a small $\varepsilon$-fraction plays some other strategy $\varphi$ should be strictly superior and else, if equally good, playing $\pi$ should do better even if an $(1-\varepsilon)$-fraction of the population plays $\varphi$. Using the terminology of adaptive dynamics, the strategy $\pi$ can invade other populations, but can never be invaded itself.\\
\indent Turning back to the static limit model of the war of attrition, we will now fit it into the abstract framework above. For for a continuum of identical players it is suitable to consider $\mathfrak{P} = ([0,1], \mathcal{B}([0,1]), m)$ where $m$ is the Lebesgue measure. Since all players chose a positive waiting time the action space is $\mathfrak{A} = (\mathbb{R}_+, \mathcal{B}(\mathbb{R}_+))$, with the usual euclidian metric. For the payoff function of a static WA with a continuum of players and increasing prize function $V \in \mathcal{C}^2[0,1]$ things turn out simplified compared to the finite model (see \cite{HaigCannings}, pp. 69). We restrict ourselves to consider mixed strategies in $\mathcal{M}_1^0(\mathbb{R}_+) \subset \mathcal{M}_1(\mathbb{R}_+)$, where $\mathcal{M}_1^0(\mathbb{R}_+)$ is the set of probability measures on $\mathbb{R}_+$ such that $d \mathbb{P}(t) = \dot{\pi}(t)dt$ and $\dot{\pi}(t) \in \mathcal{C}(\mathbb{R}_+)$. Consider a finite partition $\mathcal{Z} = \{ Z_{\varepsilon_i} \}_{i=1}^n$ of $[0,1]$ for which $m(Z_{\varepsilon_i}) = \varepsilon_i$. Let $\mathcal{A}(\mathcal{Z}) = \{ \pi_i(t)dt \}_{i=1}^n \subset \mathcal{M}_1^0(\mathbb{R}_+)$ and assume that all players in $Z_{\varepsilon_i}$ are playing $\pi_i(t)dt$ for all $i = 1, ..., n$. If $\Delta_{\mathcal{A}(\mathcal{Z})}: [0,1] \rightarrow \mathcal{M}_1^0(\mathcal{R}_+)$ is the mixed action profile corresponding to the pair $(\mathcal{Z},\mathcal{A}(\mathcal{Z}))$ we introduce the \textit{average} of $\Delta_{\mathcal{A}(\mathcal{Z})}$ as the measure
\begin{equation*}
    \bar{\Delta}_{\mathcal{A}(\mathcal{Z})}(dt) := \int_0^1 \Delta_{\mathcal{A}(\mathcal{Z})}(p) dp = \sum_{i=1}^N \varepsilon_i \pi_i(t)dt.
\end{equation*}
Note that in an $N$-player game constructed so that $P = \{ 1/N, 2/N, ..., 1 \}$, $\mathcal{P} = 2^P$ and $\mu(\{ i/N \}) = 1/N$, for $i = 1, 2, ..., N$, and in which we consider a mixed action profile $\Delta(i/N) = \pi_i \in \mathcal{M}_1^0(\mathbb{R}_+)$, the average of $\Delta$ is nothing but the mixed distribution $(\sum_{i=1}^N \pi_i)/N$. Now, considering the behaviour of $V(t)$ as players are leaving in the static WA with infinitely many players it is by continuity clear that a single player quitting will not contribute the time evolution. For $V$ to evolve at some point of time $t$, i.e. have $V'(t) > 0$, requires a positive density of players quitting at $t$. Since $V$ is differentiable by assumption all the players quitting at $t$ will collect the same prize $V(t)-t$ and therefore the dependence of order among players locally vanishes in the limit. We conclude that, given some partition $\mathcal{Z}$ and a corresponding set of mixed strategies $\mathcal{A}(\mathcal{Z})$, the payoff function of the static model WA in the continuum limit of infinitely many players can be written on interior-exterior form as
\begin{equation}
    \mathcal{J} \left(p | \Delta_{\mathcal{A}(\mathcal{Z})} \right) := \int_0^{\infty} \left[ V \left(\int_0^t \bar{\Delta}_{\mathcal{A}(\mathcal{Z})}(dx) \right) - t \right] \Delta(p)(dt).
\end{equation}
Thus, in this case the space of profile statistics is given by the space of cdf's, i.e. $Y = \{ f: f(t) = \int_0^t \nu(dx) \text{ for all } t \in \mathbb{R}_+ \text{ and some } \nu \in \mathcal{M}_1(\mathbb{R}_+) \} $, and the mixed externality is given by the map $e: \Delta_{\mathcal{A}(\mathcal{Z})} \mapsto \int_0^t\bar{\Delta}_{\mathcal{A}(\mathcal{Z})}(dx)$.\\
\indent We are now ready to start the ESS-analysis of playing $\dot{q}(t)$ in the static limit model of the WA. Considering condition (i) in Definition \ref{definition1} with the mixed action profile $\Pi_{(\dot{q}, \dot{\varphi})}^{P_{\varepsilon}}$, hence the mixed externality $e \left( \Pi_{(\dot{q}, \dot{\varphi})}^{P_{\varepsilon}} \right) = \varepsilon \varphi + (1-\varepsilon) q$, and using a Taylor expansion followed a partial integration we get that
\begin{eqnarray}
    \nonumber &&\int_0^{\infty} \mathcal{U}(t,e(\Pi_{(\dot{q}, \dot{\varphi})}^{\varepsilon}))\dot{q}(t)dt - \int_0^{\infty} \mathcal{U}(t,e(\Pi_{(\dot{q}, \dot{\varphi})}^{\varepsilon}))\dot{\varphi}(t)dt =\\
    \nonumber &&\int_0^{\infty} \left( \dot{q} - \dot{\varphi} \right) \left[ V\left( \varepsilon \varphi + (1-\varepsilon) q \right) - t \right] dt =\\
    \nonumber &&\int_0^{\infty} \left( \dot{q} - \dot{\varphi} \right) \left[ V(q) + V'(q) \left( \varepsilon \varphi -\varepsilon q \right) - t \right] dt + \mathcal{O}\left( \varepsilon ^2 \right)=\\
    \nonumber &&\varepsilon \int_0^{\infty} \left( \dot{q} - \dot{\varphi} \right) V'(q) \left( \varphi - q \right) dt + \mathcal{O}\left( \varepsilon ^2 \right)=\\
    \nonumber && -\frac{\varepsilon}{2} \int_0^{\infty} \frac{d}{dt}\left( \varphi - q \right)^2 \cdot  V'(q) dt + \mathcal{O}\left( \varepsilon ^2 \right)=\\
    &&\frac{\varepsilon}{2} \int_0^{\infty} \left( \varphi - q \right)^2  \dot{q} V''(q) dt + \mathcal{O}\left( \varepsilon ^2 \right),
    \label{difference}
\end{eqnarray}
where we have used that $V(q(t)) - t = 0$ in the third equality. Now, choosing $\varepsilon$ small enough, so that the ordo term can be neglected, we reach the conclusion that wether the first condition in Definition \ref{definition1} is fulfilled or not is depends on the geometry of the graph of $V(x)$. Since $\dot{q}(t)$ is a probability density, and by that positive, what determines the sign of the difference in payoff is the second derivative of the prize function. If $V(x)$ is convex the strategy of choosing waiting time according to $\dot{q}(t)dt$ will always constitute an ESS whereas, if instead $V(x)$ is a concave function, the $\dot{q}$-strategy can be invaded by any other strategy in $\mathcal{M}_1^0(\mathbb{R}_+)$. For a linear model in which $V(x) = kx$ for some $k > 0$, it follows by an easy calculation that there is equality in both $(i)$ and $(ii)$ in Definition \ref{definition1}. If $V(x)$ is neither convex nor concave the strategy $\dot{q}(t)dt$ is not an ESS since (\ref{difference}) could be made both positive and negative by choosing $\dot{\varphi}(t)$ carefully. In particular, if $A_- := \{ t: V''(q(t)) < 0 \}$ and $A_+ := \mathbb{R_+} \backslash A_-$, choosing $\dot{\varphi}$ so that $\varphi \equiv q$ on $A_+$ and $(\varphi(t) - q(t))^2 > 0$ on a subset of $A_-$ of positive Lebesgue measure yields a negative value in (\ref{difference}).\\
\indent According to \cite{HaigCannings}, in the static $N$-player model, a sufficient (second order variational) condition for a cdf $G_N(t)$, solving (\ref{autonomous}), to be an ESS is to have positivity in the functional
\begin{equation}
    \Upsilon_N \left[ \alpha \right] := \int_0^{\infty} Q\left[ G_N \right](t) \cdot \alpha ^2(t) dt
    \label{functional}
\end{equation}
for all functions $\alpha (t)$, i.e. $Q\left[ G_N \right](t) > 0$ , where
\begin{equation}
    Q\left[ G_N \right] = 2G_N^{N-2} + \frac{d}{dt} \left\{ \sum_{r=0}^{N-2} c_r \binom{N-2}{r} G_N^r (1-G_N)^{N-2-r} \right\}
\end{equation}
for all $t \geq 0$. In the case of a linear model, in which the sequence $\left\{ V_k \right\}_{k=1}^N$ is an increasing arithmetic progression, positivity is proven in \cite{HaigCannings}. For more general sequences the problem is left open. However, the conclusions we made from the calculations in (\ref{difference}) suggest that the positivity of $Q$ might also be true for large $N$ when the prize sequence is convex, i.e. so that the sequence of consecutive differences $\left\{ c_r \right\}_{r=0}^{N-2}$ is increasing. The next theorem implies that this is true, not only asymptotically, but for all $N \geq 2$.
\begin{theorem}
    Let $\left\{ V_k \right\}_{k=1}^N \subset \mathbb{R}_+$ be an increasing sequence such that the sequence of consecutive differences $\left\{ c_r \right\}_{r=0}^{N-2}$ is non-decreasing. Then, if $G_N(t)$ is the cdf solving (\ref{autonomous}), it holds that $Q\left[ G_N \right](t) \geq 0$ for all $t \geq 0$ and any $N \geq 2$. On the other hand, if $\left\{ c_r \right\}_{r=0}^{N-2}$ is decreasing, then $Q\left[ G_N \right](t)$ takes negative values in a set of positive Lebesgue-measure. Finally, $\lim_{t \rightarrow \infty} Q\left[ G_N \right] (t) = 1$ independently of the sequence of consecutive differences.
    \label{theorem3}
\end{theorem}
\begin{proof}
Consider the second term in the definition of $Q\left[ G_N \right]$:
\begin{align*}
    &\frac{d}{dt} \left\{ \sum_{r=0}^{N-2} c_r \binom{N-2}{r} G_N^r (1 - G_N)^{N-2-r} \right\} =\\
    &= g_N \left( N-2 \right) \left( c_{N-2} G_N^{N-3} - c_0 \left( 1-G_N \right)^{N-3} \right) +\\
    &\sum_{r=1}^{N-3} g_N c_r \binom{N-2}{r} \left[ r G_N^{r-1} \left( 1-G_N \right)^{N-2-r} - \left( N-2-r \right) G_N^r \left( 1-G_N \right)^{N-3-r} \right]\\
    &= \sum_{r=0}^{N-3} g_N G_N^r \left( 1-G_N \right)^{N-3-r} \left[ c_{r+1} \binom{N-2}{r+1} (r+1) - c_r \binom{N-2}{r} (N-2-r) \right].
\end{align*}
Each term in the sum above is positive if and only if
\begin{equation}
    \frac{c_r}{c_{r+1}} \leq \frac{\binom{N-2}{r+1} (r+1)}{\binom{N-2}{r} (N-2-r)} = 1,
    \label{increasing_cr}
\end{equation}
and since $\left\{ c_r \right\}_{r=0}^{N-2}$ is a growing sequence by assumption we get that $Q\left[ G_N \right]$ is positive for all $t \geq 0$ and all $N \geq 2$.\\
\indent If $\left\{ c_r \right\}_{r=0}^{N-3}$ is decreasing we have especially that $c_0/c_1 > 1$ and the condition in (\ref{increasing_cr}) is broken for this pair. It is then easy to check that $Q\left[ G_N \right](0) < 0$ and, because of continuity, it is therefore also negative in some neighborhood of $t = 0$.\\
\indent The last claim in the theorem follows easily from the definition of $Q$, using the explicit formula above for the term with derivative, and that $\lim_{t \rightarrow \infty} G_N(t) = 1$ and $\lim_{t \rightarrow \infty} g_N(t) = 0$.
\end{proof}
\noindent Thus, by Theorem \ref{theorem3} we get the important corollary:
\begin{corollary}
Let $\left\{ V_k \right\}_{k=1}^{N}$ be a positive and increasing sequence of numbers such that the differences $c_r = V_{r+2} - V_{r+1}$ forms a positive and increasing sequence. Further, let $G_N$ be the unique cdf solving the problem in (\ref{autonomous}) and let $g_N(t) := dG_N/dt (t)$. Then, choosing a waiting time according to the probability density $g_N(t)$ in an $N$-player static model of war of attrition, with the prize sequence $\left\{ V_k \right\}_{k=1}^{N}$, is a unique ESS.
\label{corollary1}
\end{corollary}

\noindent The conclusion of Corollary \ref{corollary1} is indeed more than what one could have hoped for since the calculation leading to (\ref{difference}) is representative only for a static model WA with a very large number of players.\\
\indent Now, switching focus towards concave prize sequences, i.e. sequences such that $\left\{ c_r \right\}_{r=0}^{N-2}$ is positive and decreasing, things turn out to be a bit more complicated. By Theorem \ref{theorem3} the functional in (\ref{functional}) can take both positive and negative values, making it a useless tool for ESS analysis\footnote{In \cite{HaigCannings}, for an $N$-player static model WA, being the $N$'th player, the functional $\Upsilon_N$ represents the second order condition of the variational problem of finding the extremum of the difference in payoff between playing the strategy $g_N(t)$ compared to any other strategy $h(t)$. Specifically, one considers the case of facing one player using $h(t)$ and the rest using $g_N(t)$.}. On the other hand, (\ref{difference}) indicates that any strategy, when used by a few players, does better against $g_N(t)$ when $N$ is large. In order to investigate this problem closer we consider the problem of playing an $N$-player static model WA in which one opponent is playing $\delta_0$, i.e. the pure strategy of quitting at $t=0$, and where the rest $N-2$ opponents play $g_N(t)$. If the expected payoff of quitting immediately is greater than the expected payoff of playing $g_N(t)$ the latter is not an ESS. Here it is important to consider a situation in which at least one of the $N-1$ opponents play a strategy different from $g_N(t)$ since if not, recalling that $g_N(t)$ solves (\ref{autonomous}), the expected payoff of playing any other strategy would be zero. According to \cite{HaigCannings} the expected payoff of playing $g_N(t)$ against one $\delta_0$-player and $N-2$ other $g_N$-players is given by the expression
\begin{align}
    \nonumber &\mathcal{J}_N\left( g_N| g_N^{\oplus(N-2)}, \delta_0 \right) =\\
    \nonumber &\int_0^{\infty} g_N(t) \left[ \sum_{r=0}^{N-2} \left( V_{r+2}-t \right) \binom{N-2}{r}G_N^r(t)(1-G_N(t))^{N-2-r} + \int_0^t G_N^{N-2}(x) dx \right] dt\\
    \nonumber &= \int_0^{\infty} g_N(t) \sum_{r=0}^{N-2} V_{r+2} \binom{N-2}{r}G_N^r(t)(1-G_N(t))^{N-2-r} dt -\\
    &\hspace{1cm}- \int_0^{\infty} \left( 1 - G_N(t) \right) \left( 1 - G_N(t)^{N-2} \right) dt,
    \label{payoff}
\end{align}
where we used partial integration and that $\int_0^t G_N^{N-2}(x)dx - t = \int_0^t \left( G_N^{N-2}(x) - 1 \right) dx$. If instead of playing $g_N(t)$ we were to play $\delta_0$, still facing the same opponents, the expected payoff is
\begin{equation*}
    \mathcal{J}_N\left( \delta_0| g_N^{\oplus(N-2)}, \delta_0 \right) = \frac{V_1 + V_2}{2}.
\end{equation*}
We denote the difference between these two numbers by $\Delta_N^{\delta_0}$. Corollary \ref{corollary1} guarantees that $\Delta_N^{\delta_0} > 0$ for all $N \geq 3$ as long as the given prize sequence is convex. By (\ref{difference}) it is natural to investigate to what extent the opposite holds true in the concave case. We start by analyzing the first term in (\ref{payoff}), assuming that $N \geq 4$. By partial integration twice, using the derivative formula from the proof of Theorem \ref{theorem3}, one finds that
\begin{align*}
    &\int_0^{\infty} g_N \sum_{r=0}^{N-2} V_{r+2} \binom{N-2}{r}G_N^r(1-G_N)^{N-2-r} dt\\
    &= V_N - \frac{N-2}{2} \left( V_N - V_{N-1} \right) +\\
    &\hspace{0.5cm}\int_0^{\infty} g_N\frac{G_N^2}{2} \sum_{r=0}^{N-4}G_N^r(1 - G_N)^{N-4-r} \binom{N-4}{r}(N-3)(N-2)(c_{r+2}-c_{r+1})dt.
\end{align*}
Expanding the binomial in the expression above, using the binomial theorem, the integral becomes explicitly solvable and one finds that
\begin{align*}
    &\int_0^{\infty} g_N \sum_{r=0}^{N-2} V_{r+2} \binom{N-2}{r}G_N^r(1-G_N)^{N-2-r} dt\\
    &= V_N - \frac{N-2}{2} \left( V_N - V_{N-1} \right) +\\
    &\hspace{0.5cm}+ \frac{(N-2)(N-3)}{2} \sum_{r=0}^{N-4}(c_{r+2}-c_{r+1}) \binom{N-4}{r} \sum_{k=0}^{N-4-r}\binom{N-4-r}{k}\frac{(-1)^k}{3+r+k}.
\end{align*}
Now, from \cite{Garrappa} we have the relation
\begin{equation}
    \sum_{k=0}^q \binom{q}{k}\frac{(-1)^k}{k-p} = \frac{q! \Gamma(-p)}{\Gamma(q+1-p)},
    \label{gammarel}
\end{equation}
for any $p \in \mathbb{R} \backslash \{ 0, 1, 2, ... \}$ and any positive integer $q$. Using (\ref{gammarel}) for simplifying the inner sum above we end up with the expression
\begin{align}
    \nonumber \Delta_N^{\delta_0} &= V_N - \frac{N-2}{2} \left( V_N - V_{N-1} \right) + \sum_{r=0}^{N-4} \frac{(r+1)(r+2)}{2(N-1)} \left( c_{r+2} - c_{r+1} \right) - \\
    \nonumber &\hspace{1cm}- \int_0^{\infty} \left( 1 - G_N \right) \left( 1 - G_N^{N-2} \right) dt - \frac{V_1 + V_2}{2}\\
    \nonumber &= \left[ V_N - \frac{N-2}{2} \left( V_N - V_{N-1} \right) + \sum_{r=0}^{N-4} \frac{r^2}{2(N-1)} \left( c_{r+2} - c_{r+1} \right) -\right. \\
    \nonumber &\hspace{1cm}- \left.\int_0^{\infty} \left( 1 - G_N \right) \left( 1 - G_N^{N-2} \right) dt + \sum_{r=0}^{N-4}\frac{3r + 2}{2(N-1)}\left( c_{r+2} - c_{r+1} \right) \right]-\\
    &\hspace{1cm}- \frac{V_1 + V_2}{2} =: A_N + C_N
    \label{ANBN}
\end{align}
While it is not easy to determine the sign of this expression for an arbitrarily finite $N$, it is possible to do this for $N$ sufficiently large. Like in previous the sections we assume that there is a fixed $\mathcal{C}^2[0,1]$-function, $V(x)$, such that $V(0) = 0$ and $V_k := V(k/N)$. Beginning with $A_N$, by Proposition \ref{proposition1} and the monotone convergence theorem we get that
\begin{align*}
    \lim_{N \rightarrow \infty} A_N = V(1) - \frac{V'(1)}{2} + \int_0^1 \frac{x^2}{2}V''(x)dx - \int_0^{V(1)} \left( 1 - V^{-1}\left( \frac{x}{V(1)} \right) \right) dx.
\end{align*}
By partial integration twice in the first integral above, and observing that
\begin{equation*}
    \int_0^{V(1)} V^{-1}\left( \frac{x}{V(1)} \right) dx = V(1) - \int_0^1 V(x)dx,
\end{equation*}
one readily finds that $\lim_{N \rightarrow \infty} A_N = 0$, independently of the sign of $V''(x)$. It is easy to see that $\lim_{N \rightarrow \infty} C_N = 0$ so that $\Delta_{\infty}^{\delta_0} = 0$. Thus, at infinity the expected payoffs are equal which is natural since the fraction of $\delta_0$-players then will be zero. What is more interesting, however, is to consider the rates of convergence in the three terms above. Let us consider prize functions of the type $V(x) = x^{\alpha}$, $\alpha \in \mathbb{R}_+$. A careful investigation of each of the terms in $A_N$ shows that both the sum and the integral converge like $\mathcal{O}(1/N)$. Hence $A_N$ converges to zero like $\mathcal{O}(1/N)$. A proof of this claim is given in Appendix \ref{appendix2}. Now, what is interesting is that the special form of $V(x)$ makes $C_N$ tends to zero like $\mathcal{O}(1/N^{\alpha})$, i.e. slower than $A_N$. Thus we conclude that $\Delta_{\infty}^{\delta_0}$ tends to zero and that $\Delta_{N}^{\delta_0} < 0$ for all $N$ large enough. We formulate this result in a theorem:
\begin{theorem}
    Let $0 < \alpha < 1$ and consider an $N$-player static war of attrition with prize sequence $\{ V_k \}_{k=1}^N$, where $V_k := \left( k/N \right) ^{\alpha}$. Then there exist an $N^* \in \mathbb{N}$ such that for all $N \geq N^*$ the strategy of choosing waiting time according to the cdf $G_N(t)$, solving \ref{autonomous}, is not an ESS and the game therefore lacks ESS strategies if the number of players exceed $N^*$.
    \label{theorem4}
\end{theorem}
\noindent Below we have included some numerical results illustrating Theorem \ref{theorem4}.
\begin{figure}[h!]
    \begin{center}
        \scalebox{0.45}{\includegraphics[trim=2.9cm 0cm 0cm 0cm]{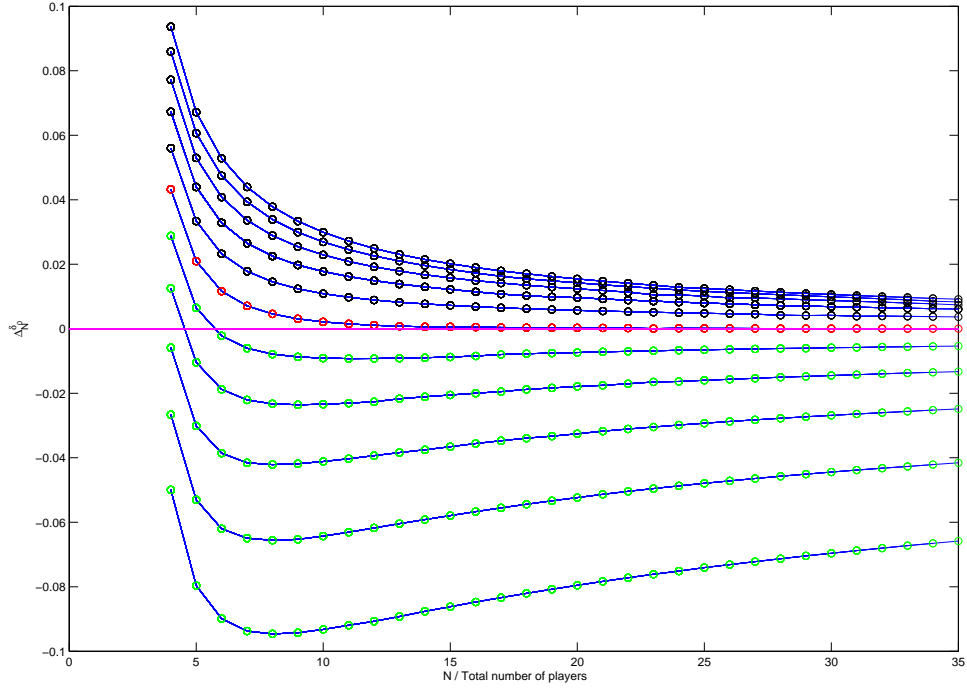}}
        \caption{Numerical results for the value of $\Delta_N^{\delta_0}$ in the range $N = 4, ..., 35$ where the prize sequence was recovered from the function $V(x) = x^{\alpha}$. The green curves, starting from below, correspond to $\alpha = 0.5, 0.6, 0.7, 0.8, 0.9$, the red curve has $\alpha = 1$, and the black curves has $\alpha = 1.1, 1.2, 1.3, 1.4, 1.5$.}
    \end{center}
\end{figure}
By \cite{HaigCannings} we note that $\Delta_N^{\delta_0} > 0$ for all $N \geq 3$ in the special case when $\alpha = 1$. It is therefore possible to push the number $N^*$ towards infinity by choosing an $\varepsilon > 0$ small enough and consider $\alpha = 1 - \varepsilon$.\\
\indent It is interesting to note that even though the time evolution of the dynamic and the static model behave the same when $N$ is large, and that Theorem \ref{theorem2} indicates that the dynamic model in a sense is static in the limit, the features of the optimal strategies differ a lot between the models. In the dynamic model the unique ESS of playing the exponential distribution with mean $(N-k)(V_{k+1} - V_k)$ in round $k$ tends to the "quasi static" strategy of playing $\dot{q}(t) = d/dt(V^{-1}(t/V(1)))$. For the static model the same result holds true by Proposition \ref{proposition1}, but if $V(x)$ is concave the limit is being reached from a sequence of non-ESS's. Thus the fundamental difference between the $N$-player dynamic model, being a repetitive game, and the $N$-player static model, being a one-shot game, is resolved only at $N = \infty$.\\
\indent For prize sequences extracted from increasing prize functions that are neither convex nor concave no definite results have been found in the $N$-player case. However, like in the concave case the calculation in (\ref{difference}) strongly indicate that a static WA with such a prize sequence lacks an ESS if $N$ is large enough.

\section*{Acknowledgements}
    The authors would kindly like to thank Torbj\"{o}rn Lundh and Philip Gerlee for introducing us to the war of attrition and for suggesting \cite{HaigCannings} as the main reference for this work.

\appendix
\section{Finding the matrix $\textbf{P}(t)$}
\label{appendix1}
In this appendix we consider the problem of finding the matrix $\textbf{P}(t)$ from Section \ref{section3}, equation (\ref{chapmankolomogorov}), in the main article. For this we consider the matrix of intensities
\begin{equation*}
    \textbf{Q} :=
    \begin{pmatrix}
        -\lambda_1 & \lambda_1 & 0 & \cdots & \cdots & 0 \\
        0 & -\lambda_2 & \lambda_2 & 0 & \cdots & 0 \\
        \vdots & \vdots & \vdots & \ddots & \ddots & \vdots \\
        0 & 0 & 0 & \cdots & -\lambda_{N-1} & \lambda_{N-1} \\
        0 & 0 & 0 & 0 & \cdots & 0
    \end{pmatrix}
\end{equation*}
where $\lambda_k \in \R_+$ for all $k = 1, ..., N-1$. The matrix $\textbf{P}(t)$ relates to $\textbf{Q}$ via the Chapman-Kolomogorov equation
\begin{equation*}
    \left\{
    \begin{array}{l}
        \frac{d\textbf{P}}{dt}(t) = \textbf{Q} \cdot \textbf{P}(t), \hspace{0.5cm} t \geq 0 \\
        \textbf{P}(0) = I.
    \end{array}
    \right.
\end{equation*}
Because of the simple bidiagonal structure of $\textbf{Q}$ it is easy to check that the eigenvalues $\left \{ \Lambda_k \right\}_{k=1}^N$ are given by $\Lambda_k = -\lambda_k$. Let $\textbf{v}_k \in \R^N$ be the right eigenvector corresponding to $\Lambda_k$ and let $\textbf{V} := \left[ \textbf{v}_1, \textbf{v}_2, ..., \textbf{v}_k \right]$. Then, investigating the eigenvector equations one by one yields
\begin{equation*}
    \textbf{V} =
    \begin{pmatrix}
        1 & \frac{\lambda_1}{\lambda_1 - \lambda_2} & \frac{\lambda_1\lambda_2}{(\lambda_1 - \lambda_3)(\lambda_2 - \lambda_3)} & \cdots & \frac{\prod_{k=1}^{N-1} \lambda_k}{\prod_{k=1}^{N-1}(\lambda_k - \lambda_N)}\\
        0 & 1 & \frac{\lambda_2}{\lambda_2 - \lambda_3} & \cdots & \frac{\prod_{k=2}^{N-1} \lambda_k}{\prod_{k=2}^{N-1}(\lambda_k - \lambda_N)}\\
        0 & 0 & 1 & \cdots & \frac{\prod_{k=3}^{N-1} \lambda_k}{\prod_{k=3}^{N-1}(\lambda_k - \lambda_N)}\\
        \vdots & \vdots & \vdots & \ddots & \vdots\\
        0 & 0 & 0 & \cdots & 1
    \end{pmatrix},
\end{equation*}
which is an invertible $N \times N$-matrix with inverse
\begin{equation*}
    \textbf{V}^{-1} =
    \begin{pmatrix}
        1 & -\frac{\lambda_1}{\lambda_1 - \lambda_2} & \frac{\lambda_1\lambda_2}{(\lambda_1 - \lambda_2)(\lambda_1 - \lambda_3)} & \cdots & (-1)^{N+1} \frac{\prod_{k=1}^{N-1} \lambda_k}{\prod_{k=1}^{N-1}(\lambda_1 - \lambda_{k+1})}\\
        0 & 1 & -\frac{\lambda_2}{\lambda_2 - \lambda_3} & \cdots & (-1)^{N+2}\frac{\prod_{k=2}^{N-1} \lambda_k}{\prod_{k=2}^{N-1}(\lambda_2 - \lambda_{k+1})}\\
        0 & 0 & 1 & \cdots & (-1)^{N+3}\frac{\prod_{k=3}^{N-1} \lambda_k}{\prod_{k=3}^{N-1}(\lambda_3 - \lambda_{k+1})}\\
        \vdots & \vdots & \vdots & \ddots & \vdots\\
        0 & 0 & 0 & \cdots & 1
    \end{pmatrix}.
\end{equation*}
Thus $\textbf{V}^{-1} \cdot \textbf{Q} \cdot \textbf{V} = \textbf{D}$, where $\textbf{D}$ is the diagonal matrix with the eigenvalues $\left \{ \Lambda_k \right\}_{k=1}^N$ on the diagonal. Defining $\textbf{A}(t) := \textbf{V}^{-1} \cdot \textbf{P}(t)$ we get the new problem
\begin{equation*}
    \left\{
    \begin{array}{l}
        \frac{d\textbf{A}}{dt}(t) = \textbf{D} \cdot \textbf{A}(t), \hspace{0.5cm} t \geq 0 \\
        \textbf{A}(0) = \textbf{V}^{-1},
    \end{array}
    \right.
\end{equation*}
where $\textbf{D} \in \R^{N \times N}$ is the diagonal matrix with the eigenvalues $\left \{ \Lambda_k \right\}_{k=1}^N$ on the diagonal. By elementary calculus one finds that $\textbf{A}(t) = \left( a_{ij}e^{-\lambda_i t} \right)$ and by the initial condition we get that
\begin{equation*}
    a_{ij} =
    \left\{
    \begin{array}{ll}
       (-1)^{i+j}\frac{\prod_{k=i}^{j-1}\lambda_k}{\prod_{k=i+1}^j (\lambda_i - \lambda_k)} , \hspace{0.5cm} & i < j \\
       1, \hspace{0.5cm} & i = j \\
       0, \hspace{0.5cm} & i > j
    \end{array}
    \right.
\end{equation*}
and since $\textbf{P}(t) = \textbf{V} \cdot \textbf{A}(t)$ we end up with an explicit expression for $\textbf{P}(t)$.

\section{Asymptotics of $A_N$}
\label{appendix2}
In this appendix we prove the claim that the terms in $A_N$ in (\ref{ANBN}) both converge like $\mathcal{O}(1/N)$. Let the underlying prize function be on the form $V(x) = x^{\alpha}$, $0 < \alpha < 1$, so that $V_k := V(k/N)$. For $A_N$ we have
\begin{align*}
    A_N = &V_N - \frac{N-2}{2} \left( V_N - V_{N-1} \right) + \sum_{r=0}^{N-4} \frac{(r+1)(r+2)}{2(N-1)} \left( c_{r+2} - c_{r+1} \right) -\\
     & \hspace{0.5cm} - \int_0^{\infty} \left( 1 - G_N \right) \left( 1 - G_N^{N-2} \right) dt,
\end{align*}
where $V_N = V(1)$ by definition. For the second term, by the mean value theorem and a Taylor approximation, we have
\begin{equation*}
    (V_N - V_{N-1})(N-2) = V'(1) + C_1/N
\end{equation*}
where $C_1$ is a constant depending on $N$ such that $|C_1| \leq \sup_{x \in [\frac{1}{2},1]} |V''(x)|/2$. For the sum in $A_N$ we have that
\begin{align}
    \sum_{r=0}^{N-4} \frac{(r+1)(r+2)}{2(N-1)} \left( c_{r+2} - c_{r+1} \right) &= \sum_{r=1}^{N-3} a_r \left( c_{r+1} - c_{r} \right)
    \label{Sumconv}
\end{align}
 and since 
 \begin{equation}
    a_{N-3}c_{N-2} - a_0c_1 = \sum_{r=1}^{N-3} a_r(c_{r+1}-c_r) + \sum_{r=1}^{N-3} (a_r - a_{r-1})c_r
    \label{trick1}
 \end{equation}
 we get the relation
 \begin{align}
    \sum_{r=0}^{N-4} \frac{(r+1)(r+2)}{2(N-1)} \left( c_{r+2} - c_{r+1} \right) = \frac{(N-3)(N-2)}{2(N-1)}c_{N-2} - \sum_{r=1}^{N-3} \frac{r}{N-1}c_r.
    \label{Sumconv1}
 \end{align}
 The first term in the equality above can be written like $V'(1)/2 + C_2/N$ where the constant $C_2 = \sup_{x \in ((N-1)/N,1)}|V''(x)$|. For the sum in the right hand side of (\ref{Sumconv1}) it holds that
 \begin{align*}
    \sum_{r=1}^{N-3} (a_r - a_{r-1})c_r = \frac{N-3}{N-1} V\left( \frac{N-1}{N} \right) - \sum_{r=1}^{N-3} \frac{1}{N-1}V\left( \frac{r+1}{N} \right)
 \end{align*} 
 where we have used the same trick as in (\ref{trick1}). Thus, collecting all the terms we get the following estimate for the sum in $A_N$
 \begin{align}
    \nonumber \sum_{r=0}^{N-4} \frac{(r+1)(r+2)}{2(N-1)} \left( c_{r+2} - c_{r+1} \right) &= \frac{V'(1)}{2} - V(1) + \sum_{r=1}^{N-3} \frac{1}{N-1}V\left( \frac{r+1}{N} \right) + \frac{C_3}{N}\\
    &\leq \frac{V'(1)}{2} - V(1) + \int_0^1 V(x) dx + \frac{C_3}{N}.
 \end{align} 
 where $C_3$ is uniformly bounded over $N$ and the last inequality follows by the concavity of $V(x)$. For the convergence of the integral term in $A_N$ we recall the following theorem (see e.g. \cite{Phillips}):
 \begin{theorem}
    If $f(x)$ is a convex function on [0,1], then
    \begin{equation*}
        \sum_{r=0}^n f\left(  \frac{r}{n} \right) \binom{n}{r}x^r(1-x)^{n-r} \geq f(x), \hspace{0.5cm} 0 \leq x \leq 1
    \end{equation*}
    for all $n \geq 1$.
    \label{convexineq}
 \end{theorem}
 \noindent Thus, by Theorem \ref{convexineq} and the fact that the approximate forward derivative of $V(x)$ is less than $V'(x)$ due to convexity, we get the following estimate
 \begin{align}
    \frac{dG_N}{dx} \leq \frac{1}{V'(G_N) - \frac{D_1}{N}},
    \label{bernsteinineq}
 \end{align}
 for $x > 0$ and $D_1 > 0$ uniformly bounded in $N$. Writing (\ref{bernsteinineq}) as a total derivative less than 1, integrating over $[0,x]$ and using the conditions that $G_N(0) = 0$ and $V(G_N(0)) = 0$ yields
 \begin{align*}
    V(G_N(x)) \leq x + \frac{D_1}{N}G_N(x)
 \end{align*}
 which in turn means that
 \begin{align*}
    G_N(x) \leq V^{-1} \left( x + \frac{D_1}{N}G_N(x) \right) &= V^{-1}(x) + \frac{G_N(x)}{V'(V^{-1}(\tilde{x}))}\frac{D_1}{N}\\
    &\leq V^{-1}(x) + \frac{D_2}{N}
 \end{align*}
 for an $\tilde{x} \in (x, G_N(x)D_1/N)$. Hence, since $G_N$ converges to $V^{-1}(x)$ uniformly (due to \cite{BuchananHildebrandt}) like $\mathcal{O}(1/N)$ so does the integral term in $A_N$.

\end{document}